\documentclass[journal,10pt,twocolumn]{IEEEtran}  

\usepackage{amssymb}
\usepackage{amsthm} 
\usepackage{multirow}
\usepackage{booktabs}
\usepackage{amsfonts}
\usepackage{ltablex}
\keepXColumns
\usepackage{wrapfig}
\usepackage{subcaption}
\usepackage{algorithmic,algorithm}
\usepackage{color, soul}
\ifCLASSINFOpdf
   \usepackage[pdftex]{graphicx}
  \graphicspath{{../pdf/}{../jpeg/}}
   \DeclareGraphicsExtensions{.pdf,.jpeg,.png}
\else
  \usepackage[dvips]{graphicx}
   \graphicspath{{../eps/}}
   \DeclareGraphicsExtensions{.eps}
\fi
\usepackage{cite}
\usepackage{amsmath}
\usepackage{array}
\usepackage{epstopdf}
\usepackage{extarrows}
\usepackage{hyperref}
\usepackage{multirow}
\usepackage{bm}  
\usepackage[table]{xcolor}
\usepackage{pifont}
\usepackage{supertabular}
\usepackage{scalerel}
\usepackage{float} 

\newcommand\AC{\ensuremath{\mathcal{A}}}

\newcommand\DC{\ensuremath{\mathcal{D}}}

\newcommand\MC{\ensuremath{\mathcal{M}}}

\newcommand\SC{\ensuremath{\mathcal{S}}}

\newcommand\xB{\ensuremath{\bm{x}}}

\newcommand\PsiB{\ensuremath{{\bm \Psi}}}

\newcommand\MF{\ensuremath{\mathbf{M}}}

\newcommand\SF{\ensuremath{\mathbf{S}}}

\newcommand\FR{\ensuremath{\mathrm{F}}}

\newcommand\IR{\ensuremath{\mathrm{I}}}

\newcommand\LR{\ensuremath{\mathrm{L}}}

\newcommand\PR{\ensuremath{\mathrm{P}}}

\newcommand\TR{\ensuremath{\mathrm{T}}}

\newcommand\WR{\ensuremath{\mathrm{W}}}

\newcommand\bR{\ensuremath{\mathrm{b}}}

\newcommand\dR{\ensuremath{\mathrm{d}}}

\newcommand\iR{\ensuremath{\mathrm{i}}}

\newcommand\mR{\ensuremath{\mathrm{m}}}

\newcommand\tR{\ensuremath{\mathrm{t}}}

\newcommand\RSS{\ensuremath{\mathrm{RSS}}}
\newcommand\OCC{\ensuremath{\mathrm{OCC}}}

\newcommand\ENS{\ensuremath{\mathrm{ENS}}}

\newcommand\sign{\ensuremath{\mathrm{sign}}}

\newcommand\sigmoid{\ensuremath{\mathrm{sigmoid}}}
\newcommand\stack{\ensuremath{\mathrm{stack}}}

\newcommand\rx{\ensuremath{\mathrm{rx}}}
\newcommand\tx{\ensuremath{\mathrm{tx}}}
\newcommand\corr{\ensuremath{\mathrm{corr}}}
\newcommand\smax{\ensuremath{\mathrm{smax}}}
\newcommand\smin{\ensuremath{\mathrm{smin}}}

\newcommand\Eb{\ensuremath{\mathbb{E}}}

\newcommand\Rb{\ensuremath{\mathbb{R}}}
\newcommand\Zb{\ensuremath{\mathbb{Z}}}

\newcommand{\wt}{\widetilde}
\newcommand{\wh}{\widehat}

\usepackage{acronym}
\newacro{isac}[ISAC]{integrated sensing and communication}
\newacro{rss}[RSS]{received signal strength}
\newacro{cgan}[CGAN]{conditional generative adversarial network}
\newacro{thz}[THz]{terahertz}
\newacro{sota}[SOTA]{state-of-the-art}
\newacro{mse}[MSE]{mean squared error}
\newacro{dl}[DL]{deep learning}

\newacro{pdf}[PDF]{probability density function}
\newacro{bs}[BS]{base station}
\newacro{awgn}[AWGN]{additive white Gaussian noise}
\newacro{ssb}[SSB]{synchronization
signal blocks}
\newacro{ssb}[SSB]{synchronization
signal blocks}
\newacro{dod}[DoD]{directions of departure}
\newacro{dnn}[DNN]{deep neural network}
\newacro{gan}[GAN]{generative adversarial network}
\newacro{cnn}[CNN]{convolutional neural network}

\newtheorem{proposition}{Proposition}

\begin{document}

\title{Advancing THz Radio Map Construction and Obstacle Sensing: An Integrated Generative Framework in ISAC
}

\author{
        Tianyu~Hu,~\IEEEmembership{Student Member,~IEEE,}
        Shuai Wang,~\IEEEmembership{Member,~IEEE,}
        Yunhang Xie,\\
        Lingxiang Li,~\IEEEmembership{Member,~IEEE,}   
        Zhi Chen,~\IEEEmembership{Senior Member,~IEEE,}
        Boyu Ning,~\IEEEmembership{Member,~IEEE,}\\
        Wassim Hamidouche,~\IEEEmembership{Senior Member,~IEEE,}
        Lina Bariah,~\IEEEmembership{Senior Member,~IEEE,}\\
        Samson Lasaulce,~\IEEEmembership{Member,~IEEE,}
        \text{Mérouane} Debbah,~\IEEEmembership{Fellow,~IEEE}
        
\thanks{This paper is an extension of our work presented at the IEEE GLOBECOM Conference in 2024~\cite{hu2024towards}}
\thanks{This work was supported in part by the National Natural Science Foundation of China (NSFC) under Grant 62271121, the National Key R\&D Program of China under Grant 2024YFE0200400 and SQ2024YFE0200402, and the National Key Laboratory of Wireless Communications Foundation under Grant 2023KP01605. (Corresponding author: Lingxiang Li)}

\thanks{T. Hu, S. Wang, Y. Xie, L. Li, Z. Chen, and B. Ning are with University of Electronic Science and Technology of China (UESTC), China. W. Hamidouche, L. Bariah, S. Lasaulce, and M. Debbah are with Khalifa University, UAE. T. Hu is also with Khalifa University, UAE. (e-mail: huty@std.uestc.edu.cn; shuaiwang@uestc.edu.cn; xieyh@std.uestc.edu.cn;  lingxiang{\_}li{\_}uestc@hotmail.com; chenzhi@uestc.edu.cn; boydning@outlook.com; \{wassim.hamidouche, lina.bariah, samson.lasaulce, merouane.debbah@ku.ac.ae\}).}}


\maketitle
\begin{abstract}
\Ac{isac} in the \ac{thz} band enables obstacle detection, which in turn facilitates efficient beam management to mitigate \ac{thz} signal blockage. Simultaneously, a \ac{thz} radio map, which captures signal propagation characteristics through the distribution of \ac{rss}, is well-suited for sensing, as it inherently contains obstacle-related information and reflects the unique properties of the \ac{thz} channel. This means that communication-assisted sensing in \ac{isac} can be effectively achieved using a \ac{thz} radio map. However, constructing a radio map presents significant challenges due to the sparse deployment of \ac{thz} sensors and their limited ability to accurately measure the \ac{rss} distribution, which directly affects obstacle sensing. In this paper, we formulate an integrated problem for the first time, leveraging the mutual enhancement between sensed obstacles and the constructed \ac{thz} radio maps. To address this challenge while improving generalization, we propose an integration framework based on a \ac{cgan}, which uncovers the manifold structure of \ac{thz} radio maps embedded with obstacle information. Furthermore, recognizing the shared environmental semantics across \ac{thz} radio maps from different beam directions, we introduce a novel voting-based sensing scheme, where obstacles are detected by aggregating votes from \ac{thz} radio maps generated by the \ac{cgan}. Simulation results demonstrate that the proposed framework outperforms non-integrated baselines in both radio map construction and obstacle sensing, achieving up to 44.3\% and 90.6\% reductions in \ac{mse}, respectively, in a real-world scenario. These results validate the effectiveness of the proposed voting-based scheme.
\end{abstract}

\begin{IEEEkeywords}
ISAC, obstacle sensing, radio map construction,  generalization ability,  THz communications.
\end{IEEEkeywords}

\IEEEpeerreviewmaketitle

\acresetall
\section{Introduction}

\IEEEPARstart{I}{n} recent years, \ac{thz} communication and \ac{isac} have emerged as key enablers of future sixth-generation (6G) wireless networks, offering advanced capabilities in terms of high-speed data transmission and environmental sensing \cite{ning2023beamforming, 9737357, 10286447, 10833779, 10243495}. The intrinsic characteristics of \ac{thz} signals, marked by narrow directional beams and weak diffraction performance, expose them to frequent obstruction by obstacles~\cite{chen2021coverage}, thereby negatively impacting the coverage and quality of \ac{thz} communication signals. Consequently, the demand for obstacle sensing with \ac{thz} \ac{isac} becomes imperative. The blockage sensitivity of \ac{thz} signals facilitates obstacle sensing, while the sensed obstacles, in turn, enable \ac{thz} communication systems to proactively mitigate link disruptions through strategic interventions such as beam switching, relaying, or hand-off~\cite{pang2023cellular, charan2021vision, 9771340, 9140329}, thereby achieving both integration and coordination gains in \ac{isac}.

\begin{figure}[!t]
\centering
\includegraphics[width=7cm,height=3cm] {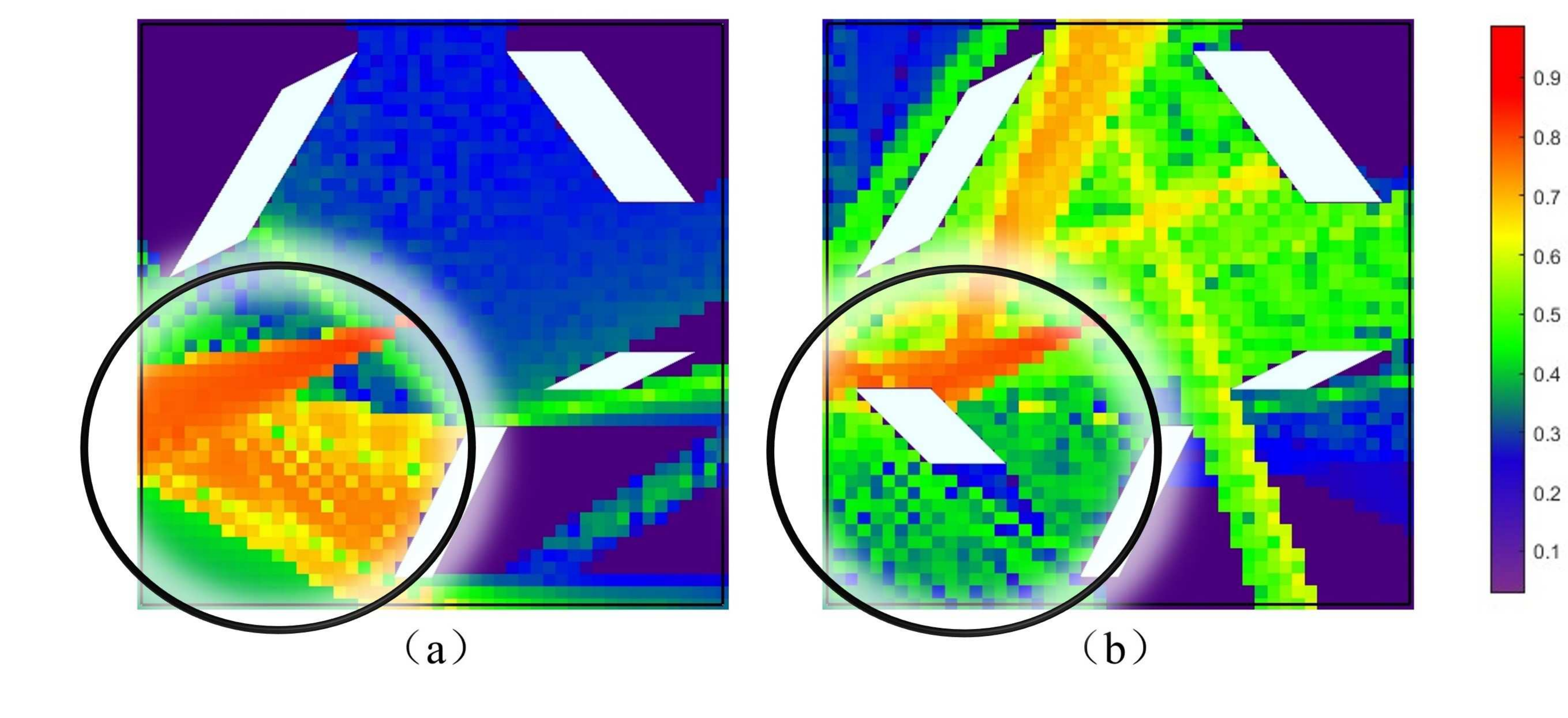}
\vspace{-10pt}
\caption{\acs{thz} radio maps with the same beam direction, where each block represents an obstacle and different colors indicate the strength of \acs{rss}. (a) Without the lower-left obstacle. (b) With the lower-left obstacle.}
\label{area_example}
\end{figure}

Radio maps, which capture channel knowledge of \ac{rss} across all transceiver locations, provide valuable insight into the propagation characteristics of wireless signals affected by the physical environment, particularly obstacles \cite{10430216,9146281,9954187,wang2024generativeairfsensing,hu2024towards,10626572}. In this context, the radio map generated during communication effectively represents the radio environment, facilitating applications such as network optimization and beam interventions. Intuitively, since the strength of \ac{rss}, affected by obstacles, provides direct access to information about the shape and location of obstacles, the radio environment characterized by \ac{rss} distribution can be utilized for obstacle sensing, enabling wireless communication systems to possess sensing capabilities, i.e., communication-assisted sensing in \ac{isac}.

Given the blockage sensitivity and directionality of \ac{thz} beams \cite{hu2024towards,10626572}, \ac{thz} radio maps are particularly well-suited for obstacle sensing. This is because the \ac{rss} strength measured in the \ac{thz} band is highly correlated with the presence of obstacles within the beam's range. Moreover, \ac{thz} radio maps that vary with beam direction share the same environmental semantics (i.e., the same obstacle layout) across different directions. Note that since \ac{thz} radio maps are determined by beam directions, we also refer to them as directional maps. This is illustrated in Fig.~\ref{area_example}, which shows such maps with directionality. From the areas outlined by the black circle in Fig.~\ref{area_example}, it can be observed that the obstacle in the lower-left corner of Fig.~\ref{area_example}(b) does not appear in Fig.~\ref{area_example}(a), demonstrating the attenuation effect of this obstacle on the \ac{thz} beam and the corresponding \ac{rss}.

However, in practice, it is challenging to fully acquire the \ac{rss} distribution, i.e., to construct a \ac{thz} radio map for obstacle sensing. This difficulty arises from the impracticality of positioning \ac{thz} sensors (e.g., user equipment) to measure \ac{rss} at every conceivable location for every scenario requiring environmental information. In other words, only a limited amount of prior knowledge of the \ac{thz} radio map can be obtained from the sparsely and randomly deployed sensors. Overcoming this limitation and deriving a comprehensive propagation model that enables obstacle sensing with good generalization (i.e., the ability to perform well in unseen scenarios or environments) remains an open research problem.

Although there have been some related studies, they have not fully addressed this issue. For radio map construction, existing approaches can typically be categorized into model-assisted \cite{8269065,9500337,9154223,10041012,zeng2022uav,10437033} and data-driven methods \cite{hu20233d,levie2021radiounet,10130091,9523765,10626572,hu2024towards}, where the former assumes that \ac{rss} follows a signal propagation model, and the latter primarily leverages the powerful data mining capabilities of \ac{dl}. Works such as \cite{levie2021radiounet,10130091,9523765} used prior physical environmental maps for radio map construction but did not account for corresponding environmental sensing. On the other hand, authors in~\cite{8269065,9500337,9154223,10041012,zeng2022uav,10437033,10626572} proposed jointly estimating a virtual environmental map and radio map. However, their generalization capabilities may be limited, as they adopt environment-specific online training and deployment. Additionally, the difference between the virtual map and the actual physical environment can be substantial. For \ac{thz}-based sensing, \cite{9449031,9967989,10061469,9729746,9838764} utilized \ac{thz} communication signals for target sensing within the \ac{isac} framework. However, such approaches may lack the capability to accurately depict the radio maps. To achieve both radio and environmental information, \cite{hu2024towards,10626572} simultaneously addressed \ac{thz} radio map construction and obstacle sensing. However, \cite{10626572} did not focus on the generalization ability, while \cite{hu2024towards} failed to meet the precision requirements for \ac{thz} beam strategic interventions in obstacle sensing.

It is important to note that \ac{thz} radio map construction and obstacle sensing are not sequential tasks, but complementary ones. A more accurate \ac{thz} radio map enables more precise sensing, while sensing results can, in turn, improve radio map construction due to the additional environmental information they provide. Moreover, since multiple directional \ac{thz} radio maps share the same environmental semantics, obstacles can be sensed based on an ensemble of multi-maps~\cite{10626572}, rather than relying on individual maps for separate sensing results~\cite{hu2024towards}. Motivated by the above, we aim to address the necessary yet challenging tasks of \ac{thz} radio map construction and obstacle sensing in an integrated manner, utilizing the shared environmental semantics. Specifically, we formulate an integrated problem to mutually improve the accuracy of construction and sensing, proposing a voting-based sensing scheme to enable ensemble and generalization. To solve this problem with enhanced generalization, and considering the powerful capabilities of generative models, we develop an integration framework based on \ac{cgan}~\cite{goodfellow2014generative,isola2017image}, balancing generation quality and inference speed. \Acp{cgan} enable efficient inference with a single forward pass under the guidance of the input condition, making them computationally suitable for real-time construction and sensing. While alternative generative models exist, their comparative analysis is beyond the scope of this work. Compared to \ac{sota} studies~\cite{8269065,9500337,9154223,10041012,zeng2022uav,10437033,hu20233d,levie2021radiounet,10130091,9523765,10626572,hu2024towards,9449031,9967989,10061469,9729746,9838764}, which suffer from limited generalization and fail to balance accurate sensing and construction, our contributions address these challenges and are summarized as follows:

\begin{figure}[!t]
\centering
\includegraphics[height=4cm,width=8.5cm] {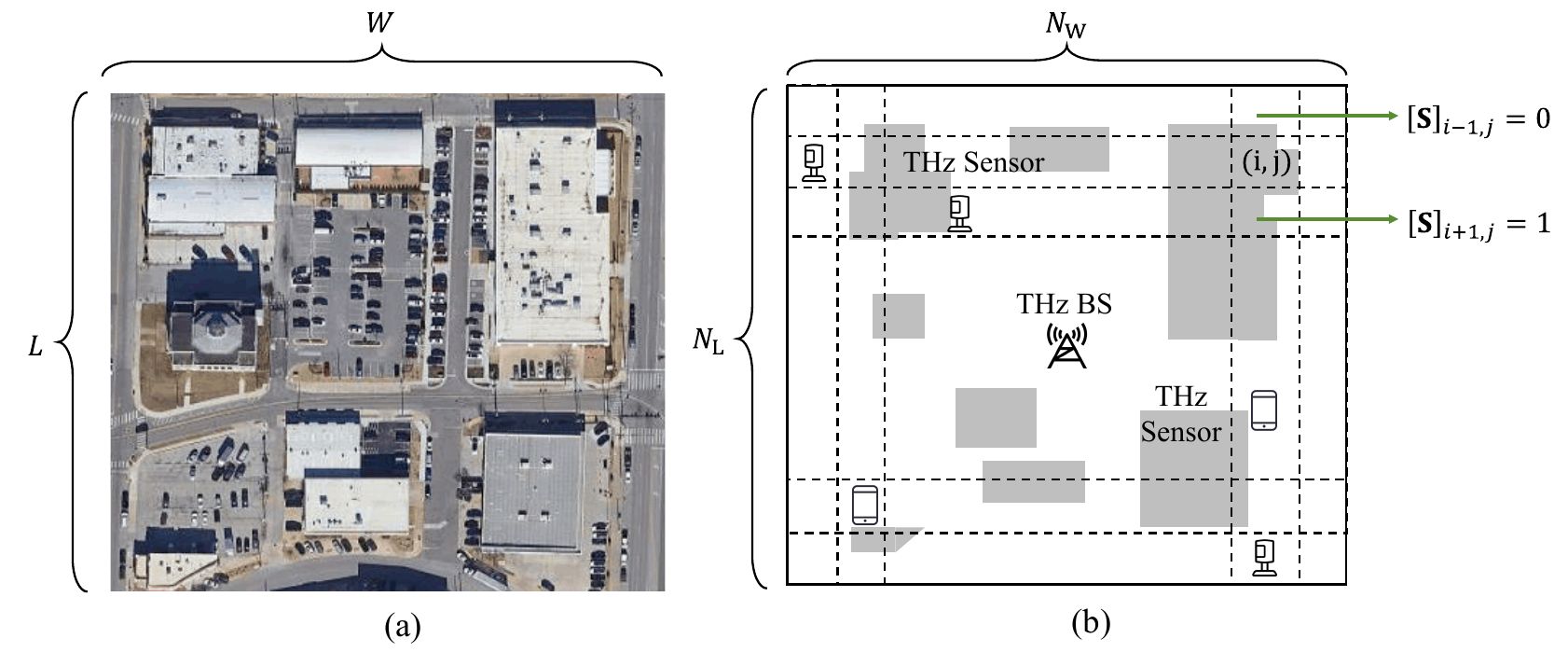}
\vspace{-5pt} 
\caption{Visual illustration of an example scenario for the considered 2D area: (a) a satellite image of a section of {\it Oklahoma}, and (b) the corresponding schematic diagram.}
\label{new_scenario}
\end{figure}

\begin{itemize}
\item We formulate an integrated problem for \ac{thz} radio map construction and obstacle sensing, which not only considers the probability distribution of obstacle layouts but also incorporates the shared environmental semantics across different directional \ac{thz} radio maps. To the best of the authors' knowledge, this is the first time these two tasks have been addressed in an integrated manner.
\item We propose a \ac{cgan}-based integration framework to solve the formulated problem and achieve \ac{thz} \ac{isac}, where the \ac{cgan} enhances generalization. By leveraging the objective function of the problem, the shared environmental semantics, and a U-Net module~\cite{ronneberger2015u} capable of uncovering the manifold nature of \ac{thz} radio maps, the \ac{cgan} effectively learns the \ac{thz} signal propagation mechanism with respect to the distribution of obstacle layouts.
\item We propose a voting-based sensing scheme that exploits the shared environmental semantics, wherein "hard" and "soft" voting strategies are applied to different directional \ac{thz} radio maps. The proposed scheme facilitates mutual accuracy improvements, thereby supporting the integrated solution to the problem within the proposed framework.
\item Simulation results show that the proposed framework significantly outperforms existing baselines~\cite{hu2024towards,hu20233d,levie2021radiounet} in terms of both \ac{mse} and average precision (AP). Specifically, for MSE, consistent accuracy improvements of up to 44.3\% and 90.6\% are achieved for the construction and sensing tasks, respectively. Through simulated and real-world city scenarios, the framework’s ability to generalize, its integration benefits, and the effectiveness of the voting mechanism have been thoroughly validated.
\end{itemize}

The remainder of this paper is organized as follows. Section~\ref{sec:system} introduces the system model and the preliminaries of directional THz radio maps. Section \ref{sec:integ} introduces the formulated problem as well as the proposed voting-based scheme. In Section \ref{sec:prop} and Section~\ref{sec:res}, the proposed \ac{cgan}-based framework is described and evaluated, respectively. Finally, Section~\ref{sec:conc} concludes the  paper. 

\renewcommand{\arraystretch}{1.5} 
\begin{table}[htbp]
\caption{List of Notations}
\label{tab:notations}
\begin{supertabular}{>{\centering\arraybackslash}m{0.14\textwidth}  
|>{\centering\arraybackslash}m{0.30\textwidth}} \toprule
\textbf{Symbol} & \textbf{Description} \\ \midrule
Tensors and Matrices & Bold capital letters \\ \midrule
$\Zb^{+}_{N}$ & The finite set of integers $\{1,\cdots,N\}$ \\ \midrule
$\|\cdot\|_0$ and $\|\cdot\|_2$ & $\ell_0$ norm and $\ell_2$ norm, respectively \\ \midrule
$\|\cdot\|_F$ & Frobenius norm \\ \midrule
$\odot$ & Hadamard product \\ \midrule

$\sign(x)$ & $\sign(x)=1$ if $x>0$; $\sign(x)=0$ otherwise\\ \hline
\multirow{2}{*}{$\stack(\SF,N)$} & Stacking $N$ copies of the matrix $\SF$ along \\ & a new axis to create a 3D tensor \\ \midrule
$\sign(x)$ & $\sign(x)=1$ if $x>0$; $\sign(x)=0$ otherwise\\ \midrule
$[\SF]_{i,j}$ & The element $(i, j)$ of the matrix $\SF$\\ \midrule
\multirow{2}{*}{$[\stack(\SF,N)]_{d}$} & The slice along the $d$-th dimension of \\ & the tensor $\stack(\SF,N)$\\ \midrule
\multirow{2}{*}{\(\mathbb{E}_{x}\left[f\left(x\right)\right]\)} & The expectation of function \(f\left(x\right)\)  \\ &  w.r.t. distribution \(x\sim p\left(x\right)\) \\ \midrule
$\nabla_{x}f(x)$ & The gradient of function \(f\left(x\right)\) w.r.t. $x$\\ 
 \bottomrule
\end{supertabular}
\end{table}

\section{Preliminaries and System Model}
\label{sec:system}
For the system setup considered, we first model the environment to be sensed, along with the corresponding \ac{thz} communication scenario. Building on this, we then elaborate on the signal propagation model and proceed to characterize the directional \ac{thz} radio maps using \ac{rss}. Given the limited prior knowledge available about the \ac{thz} radio map, we introduce a corresponding construction problem and explore the challenges associated with its exploitation in the sensing process.

\subsection{System Setup Description}

As shown in Fig.~\ref{new_scenario}, we consider a two-dimensional (2D) rectangular area, $\MC = [0;L] \times [0; W] \in \Rb^{2}$, with length $L$ and width $W$. The area contains numerous obstacles with varying attributes, such as different shapes, categories, quantities, and locations, distributed in an organized yet unpredictable manner. Examples of such environments include dense urban areas and functional indoor spaces. Note that Fig.~\ref{new_scenario}(a) shows a satellite image of a section of {\it Oklahoma}~\cite{Alkhateeb2019}, and the obstacle layout shown is intended as a representative scenario within this region. Using a spatial discretization approach~\cite{9523765}, we partition the area $\MC$ into $N_{\LR} \times N_{\WR}$ grids, as illustrated in Fig.~\ref{new_scenario}(b). The coordinates at the center of the grid $(i, j)$ are denoted as $\xB_{i,j} \in \Rb^{2}$, with each grid location corresponding to the respective center coordinates.


To model the environmental scenarios within $\MC$, we represent the environmental layout, dependent on the various attributes of the obstacles inside, using a random variable $\SC$. This variable $\SC$ follows a \ac{pdf} $p(\SC)$ defined on a sample space $\AC$, which is large due to the substantial variety of obstacle configurations within $\MC$. Consequently, this model accounts for all possible environmental scenarios, with each outcome in $\AC$ corresponding to a specific scenario characterized by the determined attributes of the obstacles~\cite{grinstead2012introduction}. Building upon an occupancy map-based model~\cite{mescheder2019occupancy} to describe the environment, we define a binary matrix $\SF \in \{0, 1\}^{N_{\LR} \times N_{\WR}}$ to represent the discretized outcome of $\SC$. In this matrix, $[\SF]_{i,j} = 1$ indicates that the grid $(i,j)$ is at least half-occupied by an obstacle (as shown in Fig.~\ref{new_scenario}(b)), whereas $[\SF]_{i,j} = 0$ signifies the absence of obstacles in that grid. For integrated \ac{thz} communication and environment sensing, we assume that a \ac{thz} \ac{bs} is positioned at the center of $\MC$ with location $\xB_{\tx}$, functioning as the transmitter equipped with a directional antenna. In addition, $K$ \ac{thz} sensors, each with omnidirectional antennas, are sparsely distributed at the centers of $K$ grids, where $K < N_{\LR} \times N_{\WR}$, typically being small. These \ac{thz} sensors, acting as receivers, are responsible for capturing the \ac{rss} measurements resulting from the signals transmitted by the \ac{thz} \ac{bs}. To model this, we introduce a mask matrix $\MF_{\rx} \in \{0, 1 \}^{N_{\LR} \times N_{\WR}}$, where an entry $[\MF_{\rx}]_{i,j} = 1$ indicates the presence of a sensor within the grid $(i,j)$, and $[\MF_{\rx}]_{i,j} = 0$ otherwise. Therefore, the mask matrix satisfies $\|\MF_{\rx}\|_0 = K$, where $\|\cdot\|_0$ denotes the $\ell_0$ norm.

\subsection{Directional Propagation Model and Radio Maps}
For a given environmental scenario with layout $\SF$, we consider a narrow \ac{thz} beam emitted from the \ac{bs}, with direction $\theta_{d}$ and beamwidth $\theta_{\bR}$, for communication. To characterize the power attenuation of the corresponding \ac{thz} signal as it propagates from the \ac{bs} location $\xB_{\tx}$ to the grid $\xB_{i,j}$, we define the function $\Gamma_{\SF}\left(\xB_{i,j}, \theta_{d}, \theta_{\bR}\right)$. This function accounts for the interaction of the signal with obstacles in $\SF$, such as reflection and scattering, along the propagation path. Assuming that \ac{thz} beams cannot penetrate obstacles, the \ac{rss} at the location $\xB_{i,j}$ resulting from the directional \ac{thz} beam is given by:
\begin{align}
\Psi_{\SF}\left(\xB_{i,j},\theta_{d},\theta_{\bR}\right)=P_{\TR}\Gamma_{\SF}\left(\xB_{i,j},\theta_{d},\theta_{\bR}\right)\left(1-[\SF]_{i,j}\right)+\sigma^{2},		
\label{Model_RSS}
\end{align} where $P_{\TR}$ is the transmit power at the \ac{bs} and $\sigma^{2}$ represents the variance of the \ac{awgn}. It is important to note that the \ac{thz} channel experiences significant challenges such as wide molecular absorption, distinct scattering properties~\cite{9514889,8761205}, and high sensitivity to obstacles, among other factors. These characteristics make the power attenuation function $\Gamma_{\SF}\left(\xB_{i,j}, \theta_{d}, \theta_{\bR}\right)$ primarily dependent on the obstacles, and obtaining an explicit expression for it is difficult. Therefore, we employ ray tracing~\cite{remcom} to characterize this power attenuation, which is further evaluated in Section~\ref{sec:res}.

The directional \ac{thz} radio map, which illustrates the distribution of the \ac{rss} for the environmental scenario with layout $\SF$, is obtained by collecting the \ac{rss} from all grid locations\footnote{Note that the use of high-gain beams prompts us to employ directional rather than omnidirectional radio maps in \ac{thz} communications, as this is more realistic in practice.}. Furthermore, by performing beam scanning, similar to the \ac{ssb} used in 5G NR, a number $N_{\dR}$ of transmitted \ac{thz} beams with a set $\DC$ of evenly distributed \ac{dod} should be considered. This leads to the generation of the corresponding multidirectional \ac{thz} radio maps. Specifically, $\DC=\{\theta_{d}=d\,\Delta\theta, d=1,\dots, N_{\dR}\}$, where $\Delta\theta$ is the angular separation between two adjacent beam \acp{dod}. To represent the $N_{\dR}$ multi-directional \ac{thz} radio maps under the obstacle layout $\SF$, we define a tensor $\PsiB_{\SF} \in \Rb^{N_{\LR}\times N_{\WR} \times N_{\dR}}$, which maps all grid locations within $\MC$ and all beam directions in $\DC$ with beamwidth $\theta_{\bR}$ to a set of corresponding signal strengths. The entry $(i, j, d)$ of this tensor is given by $\Psi_{\SF}(\xB_{i,j},\theta_{d},\theta_{\bR})$, with a particular directional radio map corresponding to the $d$-th beam direction represented as $[\PsiB_{\SF}]_{d}$.

As seen in~\eqref{Model_RSS}, the directional radio maps are inherently correlated with the environmental obstacles through which the \ac{thz} signals propagate, which motivates the use of a segmentation-based sensing approach~\cite{hu2024towards}. A simple method, for example, involves directly segmenting obstacles from $\PsiB_{\SF}$ based on the condition $\Psi_{\SF}(\xB_{i,j},\theta_{d},\theta_{\bR})=\sigma^{2}$, although this may result in the misclassification of some shadowed regions, where $P_{\TR}\Gamma_{\SF}\left(\xB_{i,j},\theta_{d},\theta_{\bR}\right)\ll\sigma^{2}$, as obstacles. Notably, multi-directional radio maps offer a view of the environment through the distribution of signal strengths from various angles, providing new opportunities to improve obstacle sensing. This concept has been explored in~\cite{10626572}, where obstacles are detected through the intersection of multi-directional \ac{thz} radio maps. Therefore, when utilizing multi-directional maps for the target scenario, it becomes possible to infer the characteristics of environmental obstacles and achieve accurate sensing.

\subsection{Constructing Directional THz Radio Map from Sparse RSS measurements for Obstacle Sensing}
Due to the limited deployment of \ac{thz} sensors within $\MC$, only a sparse set of \ac{rss} measurements, represented by the sparse tensor $\wt{\PsiB}_{\SF}=\PsiB_{\SF} \odot \stack(\MF_{\rx},N_{\dR})$, can be obtained. This tensor $\wt{\PsiB}_{\SF}$ captures the \ac{rss} with respect to $N{\dR}$ beam directions, but it does not fully characterize the complete directional \ac{thz} radio maps. In this paper, we focus on constructing the multi-directional \ac{thz} radio maps $\PsiB_{\SF}$ and accurately sensing the obstacles $\SF$ for any given environmental scenario within $\MC$, using only the \ac{rss} measurements $\wt{\PsiB}_{\SF}$ obtained in the same scenario.


Firstly, we define a function $f_{\RSS}: \Rb^{N_{\LR}\times N_{\WR}\times N_{\dR}} \to \Rb^{N_{\LR}\times N_{\WR}\times N_{\dR}}$, which takes the \ac{rss} measurement $\wt{\PsiB}_{\SF}$ as input and produces a tensor $\wh{\PsiB}_{\SF}$ w.r.t. the \ac{rss} estimations of all grid locations and all beam directions. Therefore, the task of directional \ac{thz} radio map construction is to derive $f_{\RSS}$ via solving the following optimization problem: 
\begin{subequations}\label{construction}
\begin{align}
	\underset{f_{\RSS}}{\mathrm{minimize}}~~&\Eb_{\SF}\left[\|\PsiB_{\SF}-\wh{\PsiB}_{\SF}\|_{\FR}^{2}\right], \label{const:2a} \\
	 \text{s.t.}~~& 	\wh{\PsiB}_{\SF}\odot\stack(\MF_{\rx},N_{\dR})=\wt{\PsiB}_{\SF},\, \forall\SF, \label{const:2b}  \\
     & \wh{\PsiB}_{\SF}=f_{\RSS}(\wt{\PsiB}_{\SF}),\forall\SF. \label{const:2c}
\end{align} 
\end{subequations} Given the limited prior data $\{\Psi_{\SF}(\xB_{i_k,j_k},\theta_{d},\theta_{\bR})\}_{k=1}^{K}$, the constraint~\eqref{const:2b} is introduced, and the optimal solution $f_{\RSS}^\star$ to problem~\eqref{construction} minimizes the \ac{rss} estimation error across all possible environmental scenarios, thereby approximating the directional radio maps $\PsiB_{\SF}$, $\forall \SF \in \AC$, to the maximum extent.

Recall that conventional model-assisted approaches for radio map construction rely on accurate modeling of the signal propagation process. However, achieving this accuracy is particularly challenging for constructing \ac{thz} radio maps. The unique characteristics of \ac{thz} signal propagation, as described in~\eqref{Model_RSS}, and the corresponding radio maps introduce several difficulties: 1) Uniqueness of the \ac{thz} channel: Compared with statistical sub-6G channels, the \ac{thz} channel's power attenuation function $\Gamma_{\SF}\left(\xB_{i,j},\theta_{d},\theta_{\bR}\right)$ is more complex and heavily influenced by obstacles, 2) Directionality of \ac{thz} radio maps: \ac{thz} radio maps depend not only on the environmental layout but are also significantly affected by the beam directions and beamwidth; 3) Complex environmental layouts: The complex and often inaccessible nature of environmental layouts, along with their intricate interaction with signal propagation in $\Gamma_{\SF}\left(\xB_{i,j},\theta_{d},\theta_{\bR}\right)$, adds further complexity. Existing model-based approaches \cite{8269065,9500337,9154223,10041012,zeng2022uav,10437033}, which are explicitly designed for the sub-6G band, may lead to model mismatches when applied to \ac{thz} radio maps, significantly undermining their effectiveness.

In contrast, data-driven approaches hold promise in circumventing the model mismatch issue. These approaches yield desirable performance due to \ac{dl}'s strong capability for complex function approximation~\cite{hu2024towards, hu20233d, levie2021radiounet, 10130091, 9523765}. The key idea is to directly learn the complex signal propagation process $\Gamma_{\SF}\left(\xB_{i,j}, \theta_{d}, \theta_{\bR}\right)$ from historical \ac{rss} measurements across a variety of environmental scenarios in $\AC$, and then use this knowledge to infer the radio maps for a target environmental scenario. Building on this idea, a \ac{dnn} with a carefully designed structure can be trained to address the intractable functional optimization in~\eqref{construction}. This is achieved through standard \ac{dnn} training on a set of historical measurements from different environmental scenarios. Consequently, the solution $f_{\RSS}$ is represented by the trained \ac{dnn}. Specifically, for a given target environmental scenario, the directional \ac{thz} radio maps are estimated and can subsequently be used for obstacle sensing tasks such as segmentation~\cite{hu2024towards} and/or intersection detection~\cite{10626572}. While this approach is promising, a significant challenge remains due to issues with poor generalization. This problem arises for two main reasons. First, there are typically only a limited number of \ac{rss} measurements available, which results in a large solution space for problem~\eqref{construction}. This makes it difficult to identify a solution with good generalization. In practice, to improve generalization performance, a prohibitively large number of historical data samples would be required for training, significantly increasing the training cost. Second, the sample space $\AC$ of $\SF$ in problem~\eqref{construction} is too vast to be fully covered by the available historical data. The target environmental scenario often lies outside the training distribution, meaning that it is not similar to any of the historical training samples. As a result, the trained \ac{dnn} performs poorly on the target environmental scenario, leading to inaccurate radio map estimation and poor obstacle sensing performance~\cite{9693274,goodfellow2016deep}.

\section{Integrated Directional Radio Map Construction and Obstacle Sensing}
\label{sec:integ}
In this section, to address the challenge of poor generalization, we propose an integrated paradigm for constructing directional radio maps and sensing obstacles, inspired by data-driven approaches. This approach is motivated by the observation that the construction of multi-directional \ac{thz} radio maps and the obstacle sensing process should mutually support each other. In essence, as discussed earlier, the accurate construction of multi-directional \ac{thz} radio maps, which share the same environmental semantics, should lead to more precise obstacle sensing. Additionally, leveraging prior information about sensed obstacles is crucial for mitigating errors in the construction of multi-directional \ac{thz} radio maps. Therefore, rather than relying solely on a data-driven map construction method, which faces challenges related to poor generalization, we integrate the tasks of constructing multi-directional \ac{thz} radio maps and obstacle sensing into a unified problem formulation. This integrated approach is then addressed using sophisticated \ac{gan}-based deep learning techniques. It is important to note that these two novel elements, the integration of tasks and the use of \acp{gan}, effectively mitigate the issue of poor generalization (which will be further elaborated in Section~\ref{sec:prop}), thereby enhancing the performance of both tasks. To further substantiate the motivation for using multi-directional \ac{thz} radio maps in obstacle sensing, we adopt the concept of ensemble learning~\cite{zhou2012ensemble} to validate our intuition. In addition, we reformulate the problem~\eqref{construction} as a guide for the subsequent design of the \ac{gan} model. It is worth noting that more precise obstacle information can assist the \ac{gan} model in generating more accurate radio maps, which, in turn, enhances the performance of the ensemble learning-inspired obstacle sensing, thereby achieving the integration goal.

\subsection{Generic Problem and Design Considerations}

Let us define $f_{\OCC}: \Rb^{N_{\LR}\times N_{\WR}\times N_{\dR}} \to \{0,1\}^{N_{\LR}\times N_{\WR}}$ as a function that maps the estimated multi-directional \ac{thz} radio maps to an estimated occupancy map $\hat{\SF}$ for any given environmental scenario. The general problem of integrated \ac{thz} radio map construction and obstacle sensing is then formulated as:
\begin{subequations}\label{reconstruction_problem}
\begin{align}
    \underset{f_{\RSS},\,f_{\OCC}}{\mathrm{minimize}}~~&\Eb_{\SF}\left[\|\PsiB_{\SF}-\wh{\PsiB}_{\SF}\|_{\FR}^{2}+ C(\SF, \hat{\SF})\right], \label{const:3a} \\
    \text{s.t.}~~& \hat{\SF}=f_{\OCC}(\wh{\PsiB}_{\SF}),\forall\SF, \label{const:3b}\\
    & \eqref{const:2b}\eqref{const:2c}
\end{align} 
\end{subequations}
The equality constraint~\eqref{const:3b} represents the process of obstacle sensing derived from multi-directional \ac{thz} radio maps, while $C(\SF, \hat{\SF})$ is a loss function that quantifies the dissimilarity between the two inputs. The generic problem formulation~\eqref{reconstruction_problem} captures the integrated tasks of \ac{thz} radio map construction and obstacle sensing due to the inclusion of the equality constraint~\eqref{const:3b} and the loss $C(\SF, \hat{\SF})$. These two elements work together to impose prior knowledge of environmental obstacles $\SF$ on the constructed multi-directional \ac{thz} radio maps $\wh{\PsiB}_{\SF}$, thereby making the two tasks complementary and mutually reinforcing.

Intuitively, for any environmental scenario $\SF$ within $\AC$, if the function $f_{\RSS}$ can generate the correct multi-directional \ac{thz} radio maps $\wh{\PsiB}_{\SF}$ from sparse historical measurements $\wt{\PsiB}_{\SF}$, then $\PsiB_{\SF}$ will be perfectly recovered, and the obstacle sensing task will succeed as well~\cite{hu2024towards, 10626572}. In practice, however, the construction of multi-directional \ac{thz}  radio maps will inevitably be imperfect due to the challenge of accurately learning the complex \ac{thz}  signal propagation process under the obstacle layout distribution $p\left(\SC\right)$. This insufficient understanding of the propagation process prevents conventional data-driven approaches from effectively generalizing to unseen scenarios. In contrast, the loss $C(\SF, \hat{\SF})$, introduced in the integrated problem, penalizes errors in obstacle sensing that stem from the construction errors. This compels the function $f_{\RSS}$ to generate multi-directional \ac{thz} radio maps $\wh{\PsiB}_{\SF}$ that contain accurate environmental information. As a result, the integration of these tasks allows for a more comprehensive and refined characterization of the \ac{thz} propagation process from the perspective of obstacles, thereby improving generalization and achieving greater accuracy in both $\wh{\PsiB}_{\SF}$ and $\hat{\SF}$.

Although the formulation of problem~\eqref{reconstruction_problem} is well-motivated, there are still two unresolved issues. The first issue concerns the instantiation of function $f_{\OCC}$, as the accuracy requirements for such mutual tasks are not fully clear. Furthermore, noise-driven errors complicate the integration of these tasks. Specifically, poor construction of $\wh{\PsiB}_{\SF}$ may hinder obstacle sensing, while erroneous sensing results can degrade the accuracy of the directional \ac{thz} radio maps, ultimately impacting generalization. Moreover, environmental noise can obscure the distinction between obstacles and regions where $P_{\TR}\Gamma_{\SF}\left(\xB_{i,j},\theta_{d},\theta_{\bR}\right)\ll\sigma^{2}$, which could lead to sensing errors in $f_{\OCC}$, even when  ideal maps $\PsiB_{\SF}$ are used. Therefore, $f_{\OCC}$ must be designed to account for these interdependencies, mitigating the negative effects of integration errors and ensuring reliable sensing performance. Without this consideration, the solution to~\eqref{reconstruction_problem} may not be satisfactory. To address this challenge, we incorporate a voting-based sensing solution from ensemble learning~\cite{zhou2012ensemble} into the design of $f_{\OCC}$, which will be further elaborated later. The second issue pertains to the formulation of the loss function $C(\SF, \hat{\SF})$. It must recognize that obstacle sensing is fundamentally a classification task for each grid while also considering the mutual accuracy requirements of both tasks. To this end, we utilize cross-entropy to define $C(\SF, \hat{\SF})$, which will be analyzed in Section~\ref{subsec:intr}.

\subsection{Voting-based Obstacle Sensing}
\label{subsec:vote}
By aggregating the outputs of several well-trained learners with sufficient diversity, ensemble learning~\cite{zhou2012ensemble} can improve overall performance and mitigate the potential weaknesses of individual learners. Inspired by this principle, we propose a voting-based sensing scheme, which defines two key functions: an intermediate function$f_{\iR}: \Rb^{N_{\LR}\times N_{\WR}} \to \{0,1\}^{N_{\LR}\times N_{\WR}}$ and a voting strategy function  $f_{\ENS}: \{0,1\}^{N_{\LR}\times N_{\WR}\times N_{\dR}} \to \{0,1\}^{N_{\LR}\times N_{\WR}}$, within the overall sensing function \( f_{\OCC} \). Specifically, the proposed scheme assumes that the function \( f_{\iR} \) generates \( N_{\dR} \) preliminary sensing results \( \{\hat{\SF}_{\theta_{d}}\}_{\theta_{d} \in \DC} \), corresponding to individual directional \ac{thz} radio maps \( \{[\wh{\PsiB}_{\SF}]_{d}\}_{d=1}^{N_{\dR}} \), where \( f_{\iR}([\wh{\PsiB}_{\SF}]_{d}) = \hat{\SF}_{\theta_{d}} \). The function \( f_{\ENS} \) then aggregates these results \( \{\hat{\SF}_{\theta_{d}}\}_{\theta_{d} \in \DC} \) and produces a more accurate sensing \( \hat{\SF} \), i.e., \( f_{\ENS}(\{\hat{\SF}_{\theta_{d}}\}_{\theta_{d} \in \DC}) = \hat{\SF} \). Therefore, the overall sensing function \( f_{\OCC} \) can be redefined as \( f_{\OCC}(\wh{\PsiB}_{\SF}) = f_{\ENS}(\{f_{\iR}([\wh{\PsiB}_{\SF}]_{d})\}_{d=1}^{N_{\dR}}) \). Moreover, through the ensemble approach, the function \( f_{\ENS} \) may mitigate the negative effects that arise from or are imposed on the \ac{thz} radio maps during the integration process. Notably, the motivation for the design of function \( f_{\iR} \) stems from the blockage sensitivity and directionality of \ac{thz} beams \cite{hu2024towards}, while the rationale behind \( f_{\ENS} \) arises from the shared environmental semantics across the directional \ac{thz} radio maps of different beam directions~\cite{10626572}.

Naturally, the preliminary occupancy estimations \( \{[\hat{\SF}_{\theta_{d}}]_{i,j}\}_{\theta_{d} \in \DC} \), corresponding to the directional \ac{thz} radio maps \( \{[\wh{\PsiB}_{\SF}]_{d}\}_{d=1}^{N_{\dR}} \), can be interpreted as \( N_{\dR} \) votes regarding the occupied or unoccupied state of the grid cell \( (i, j) \). Consequently, the primary challenge lies in the design of the voting strategy function \( f_{\ENS} \) to effectively aggregate these \( N_{\dR} \) votes. In ensemble learning, it is crucial that the learners are well-trained to ensure effective aggregation of their outputs. Similarly, before designing the function \( f_{\ENS} \), we assume that the sensing performance of the function \( f_{\iR} \) with respect to each directional \ac{thz} radio map is superior to that of a blind guessing estimator. That is, we assume the following error probability for the preliminary occupancy estimation:
\begin{equation}
\begin{aligned}
P\left([\hat{\SF}_{\theta_{d}}]_{i,j}\neq [\SF]_{i,j}\right)=\epsilon<0.5, \forall \theta_{d}\in\DC, i\in\Zb^{+}_{N_{\LR}}, j\in\Zb^{+}_{N_{\WR}}, 
 \label{sub_error}
\end{aligned}
\end{equation} 
where  $\epsilon$ represents the error probability of the preliminary estimation.
\begin{proposition}
If a majority voting strategy is employed (i.e., the state with more than half of the votes is selected as the final result \( [\hat{\SF}]_{i,j} \) of the function \( f_{\OCC} \)), and the error probabilities of \( \{[\hat{\SF}_{\theta_{d}}]_{i,j}\}_{\theta_{d} \in \DC} \) are assumed to be independent, the error probability of \( [\hat{\SF}]_{i,j} \) is given by:
\begin{equation}
\begin{aligned}
P\left([\hat{\SF}]_{i,j}\neq [\SF]_{i,j} \right)\leq e^{-2 N_{\dR} \left(\frac{1}{2}-\epsilon\right)^{2}}.
 \label{ensemble_inequality}
\end{aligned}
\end{equation}
\end{proposition}
\begin{proof}
See Appendix \ref{Annex:a}, following \cite{zhou2012ensemble}.
\end{proof}

It follows from \eqref{ensemble_inequality} that as the number \( N_{\dR} \) of beam directions increases, the error probability of the aggregated sensing result \( [\hat{\SF}]_{i,j} \) via voting decreases exponentially and eventually saturates to zero. Expanding \eqref{ensemble_inequality} from the grid cell \( (i, j) \) to all grids in the region \( \MC \), we can conclude that leveraging \( N_{\dR} \) directional \ac{thz} radio maps \( \{[\wh{\PsiB}_{\SF}]_{d}\}_{d=1}^{N_{\dR}} \) is advantageous for reducing the error in obstacle sensing. Even if the quality of \( \{[\wh{\PsiB}_{\SF}]_{d}\}_{d=1}^{N_{\dR}} \) is suboptimal, as long as the average error probability of \( \wh{\PsiB}_{\SF} \) is less than 0.5, the obstacle sensing accuracy can still be maintained through the voting strategy. This approach also mitigates the risk that poor sensing quality would degrade subsequent construction performance. To ensure reliability, we multiply \( C(\SF, \hat{\SF}) \) by an indicator \( \eta \), which reflects the preliminary sensing performance, where \( \eta = 1 \) if the error probability \( \epsilon < 0.5 \) is achievable, and \( \eta = 0 \) otherwise. In this case, voting-based sensing and integration are only performed when \( \eta = 1 \).

However, majority voting with directional \ac{thz} radio maps may not be suitable for obstacle sensing, as it neglects the relationship between obstacles and \ac{rss}. For instance, with a large obstacle, most beam directions may cause a region to be shadowed in the radio maps, leading to its misclassification as part of the obstacle during majority voting, even if the region is not actually occupied. To improve the credibility of the voting results, we adopt the voting strategy from \cite{10626572}, where a grid is considered occupied only if it receives unanimous votes \( \{[\hat{\SF}_{\theta_{d}}]_{i,j}\}_{\theta_{d} \in \DC} \) indicating occupation. Specifically,
\begin{equation}
\begin{aligned}
\hat{\SF} = f_{\ENS}(\{\hat{\SF}_{\theta_{d}}\}_{\theta_{d} \in \DC}) = \hat{\SF}_{\theta_{1}} \odot \hat{\SF}_{\theta_{2}} \odot \cdots \odot \hat{\SF}_{\theta_{N_{\dR}}}.
\end{aligned}
\label{hard_vote}
\end{equation}

To verify the validity of the insight obtained from \eqref{ensemble_inequality} under the voting strategy described in \eqref{hard_vote}, we conduct a simple experiment, as shown in Fig.~\ref{GT_beam}, using \( N' \) obstacle layouts and the corresponding \ac{thz} radio maps. In this experiment, \( [\hat{\SF}_{\theta_{d}}]_{i,j} = 1 \) if \( [\PsiB_{\SF}]_{i,j,d} \leq \sigma^2 + \varepsilon \) and \( [\hat{\SF}_{\theta_{d}}]_{i,j} = 0 \) otherwise, where \( 0 < \varepsilon \ll 1 \). Fig.~\ref{GT_beam} illustrates the \ac{mse}, defined as \( \frac{1}{N'} \sum \| \SF - \hat{\SF} \|_{\FR}^2 \), as a function of the number of preliminary occupancy estimations \( N_{\dR} \) (i.e., the number of beam directions). The results show that the \ac{mse} decreases as \( N_{\dR} \) increases. This trend aligns with the theoretical insight, despite the presence of one or two outliers due to the inherent randomness in beam directions.

\begin{figure}[!t]
\center
\includegraphics[height=3.5cm, width=8.5cm] {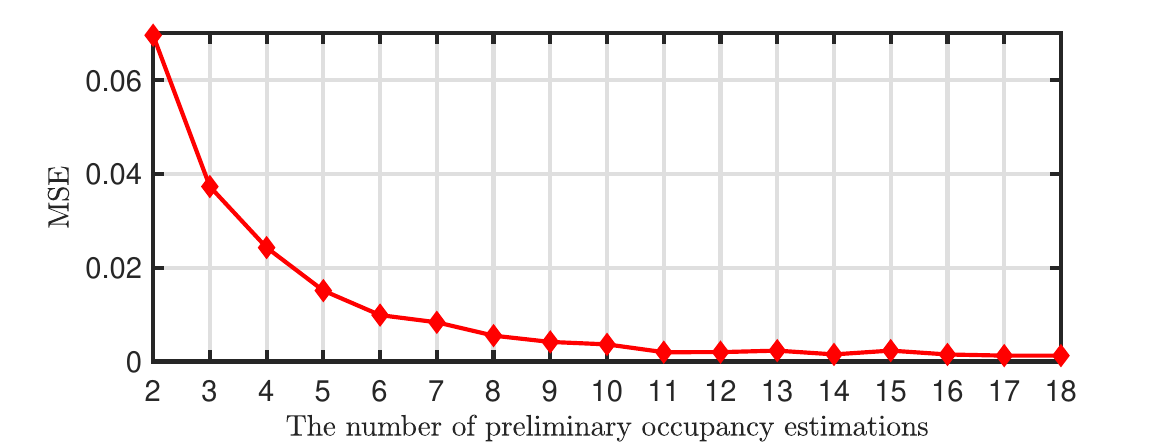}
\caption{MSE as a function of the number of preliminary occupancy estimations.}
\label{GT_beam}
\end{figure}

Nevertheless, the "hard" voting strategy in \eqref{hard_vote}, along with the discrete sensing results \( [\hat{\SF}]_{i,j} \in \{0, 1\} \) and \( [\hat{\SF}_{\theta_{d}}]_{i,j} \in \{0, 1\} \), is inherently non-differentiable, which may lead to discontinuous gradients. In other words, the applicability of \eqref{hard_vote} and the function \( f_{\OCC} \) in gradient-based optimization methods for solving \eqref{reconstruction_problem} is limited, as these methods rely on smooth and well-defined gradients for backpropagation. To address this challenge, we adopt a scaling method for preprocessing and a segmentation approach inspired by \cite{hu2024towards}, where \( [{\PsiB}_{\SF}']_{i,j,d} = 1 \) if the grid is occupied by obstacles, and \( [{\PsiB}_{\SF}']_{i,j,d} \in [\Psi_{\smin}, \Psi_{\smax}] \) represents the scaled \ac{rss} in the case of unoccupied regions. The occupancy estimation can then be obtained using the following rule:

\begin{equation}
\begin{aligned}
 \,[\hat{\SF}]_{i,j} = \begin{cases}
0, & \quad 0 \leq [{\wh{\PsiB}}_{\SF}']_{i,j,d} \leq \Psi_{\smax} \\
1, & \quad \Psi_{\smax} < [{\wh{\PsiB}}_{\SF}']_{i,j,d} \leq 1
\end{cases}
\label{occupancy_probability}
\end{aligned}
\end{equation}

Here, when \( [{\wh{\PsiB}}_{\SF}']_{i,j,d} \in [\Psi_{\smin}, \Psi_{\smax}] \), it represents the scaled estimated \ac{rss} for unoccupied areas. Note that \( \Psi_{\smax} \) and \( \Psi_{\smin} \) are predefined thresholds, with \( \Psi_{\smax} \) specifically used for obstacle segmentation. Although the effectiveness of \eqref{occupancy_probability} has been evaluated in non-integrated scenarios \cite{hu2024towards}, it remains non-differentiable.

To overcome the aforementioned issues, we employ differentiable approximations by implementing a "soft" unanimous voting strategy and introducing probabilistic values. By allowing \( {\PsiB}_{\SF} \) to undergo the same preprocessing as \( {\PsiB}_{\SF}' \) and ensuring that the function \( f_{\RSS} \) produces the \ac{rss} estimation in the same form as \eqref{occupancy_probability} in~\cite{hu2024towards}, we directly interpret the preliminary sensing results \( \{ \hat{\SF}_{\theta_{d}} \}_{\theta_{d} \in \DC} \) as \( \{ [\wh{\PsiB}_{\SF}]_{d} \}_{d=1}^{N_{\dR}} \), where \( [\wh{\PsiB}_{\SF}]_{i,j,d} \in [0, 1] \), and function \( f_{\iR} \) acts as the identity function. We then redefine the function \( f_{\ENS}: [0,1]^{N_{\LR} \times N_{\WR} \times N_{\dR}} \to [0,1]^{N_{\LR} \times N_{\WR}} \), and the "soft" voting strategy is formulated as follows:

\begin{equation}
\begin{aligned}
\hat{\SF} = f_{\ENS}(\{ \hat{\SF}_{\theta_{d}} \}_{\theta_{d} \in \DC}) = \max\left( 0, \frac{1}{N_{\dR}} \sum_{\theta_{d} \in \DC} \hat{\SF}_{\theta_{d}} - \Psi_{\smax} \right).
\label{soft_vote}
\end{aligned}
\end{equation}

According to \eqref{soft_vote}, when \( [\hat{\SF}]_{i,j} = 0 \), it indicates that the grid is unoccupied by obstacles. In contrast, when \( [\hat{\SF}]_{i,j} > 0 \), it means occupancy, with the magnitude reflecting the confidence level of the occupancy prediction. A higher value indicates greater certainty, allowing for a more nuanced representation of categorical information within the probability space. From \eqref{soft_vote}, it is evident that "soft" unanimous voting for obstacle occupancy is achieved through averaging and segmentation. As a result, the function \( f_{\OCC} \) can be simply represented as:
\[
f_{\OCC}(\wh{\PsiB}_{\SF}) = f_{\ENS}(\wh{\PsiB}_{\SF}) = f_{\ENS}(f_{\RSS}(\wt{\PsiB}_{\SF})) = \hat{\SF}.
\]
This approach enables differentiable obstacle sensing, allowing the use of gradient-based optimization techniques in subsequent processing tasks.

\subsection{Integrated Problem Formulation}
\label{subsec:intr}
As the preprocessing method from~\cite{hu2024towards} is employed, we leverage an additional technique from \cite{hu2024towards} to improve the integration of \ac{thz} radio map construction and obstacle sensing. This technique involves using $\PsiB_{\SF}$ as an extra weight tensor, which is multiplied by $(\PsiB_{\SF}-\wh{\PsiB}_{\SF})$. In this approach, the highest weight (i.e., $1$) is assigned to grid cells occupied by obstacles, while weights in the range $[\Psi_{\smin},\Psi_{\smax}]$, determined by the \ac{rss} value, are allocated accordingly. This reflects the fact that \ac{isac} systems typically prioritize areas with obstacles and hotspot regions exhibiting higher \ac{rss} values when exploiting radio environments. Building upon this, and considering the voting-based sensing scheme described in~\eqref{soft_vote} and the use of cross-entropy for dissimilarity quantification, problem \eqref{reconstruction_problem} can be reformulated as:

\begin{subequations}\label{final_problem}
\begin{align}
    \underset{f_{\RSS},\,f_{\OCC}}{\mathrm{minimize}}~&\Eb_{\SF}\left[B(\PsiB_{\SF},\wh{\PsiB}_{\SF})+ \eta C(\SF, \hat{\SF})\right]\!, \label{final_problem_obj} \\
    \text{s.t.}~& \hat{\SF}=f_{\OCC}(\wh{\PsiB}_{\SF})=f_{\ENS}(\wh{\PsiB}_{\SF}),\forall\SF, \label{const:9b}\\
    & \eqref{const:2b}\eqref{const:2c}
\end{align} 
\end{subequations}

Here, $B(\PsiB_{\SF},\wh{\PsiB}_{\SF})=\|\PsiB_{\SF}\odot(\PsiB_{\SF}-\wh{\PsiB}_{\SF})\|_{\FR}^{2}$ and $C(\SF, \hat{\SF})=N_{\dR}\sum_{i\in\Zb^{+}_{N_{\LR}}}\sum_{j\in\Zb^{+}_{N_{\WR}}}\left([\SF]_{i,j}\log[\hat{\SF}]_{i,j}+\left(1-[\SF]_{i,j}\right)\log\big(1-\right.\\\left.[\hat{\SF}]_{i,j}\big)\right)$. 

From~\eqref{final_problem}, it is evident that our integration problem is essentially a weighted combination of regression and classification tasks, which introduces a higher level of complexity. Specifically, the intractable optimization of problem \eqref{final_problem} presents challenges due to the interaction between these two tasks. The optimization of \ac{dnn} models must strike a careful balance between their competing and/or complementary objectives. Furthermore, the models must learn effective representations that perform well across both tasks while avoiding overfitting. In particular, \ac{mse} is well-suited for \ac{thz} radio map construction due to its alignment with continuous \ac{rss} outputs, while cross-entropy is more appropriate for obstacle senssing as it quantifies the divergence between predicted and true probability distributions with respect to occupancy. Additionally, $\eta$ serves not only as a weight to balance the sensing and construction tasks but also to mitigate the negative effects of integration errors, addressing the error probability $\epsilon$ in the voting strategy. This allows $C(\SF, \hat{\SF})$ to effectively contribute to the overall optimization.

\begin{figure*}[!t]
\centering
\includegraphics[width=\linewidth] {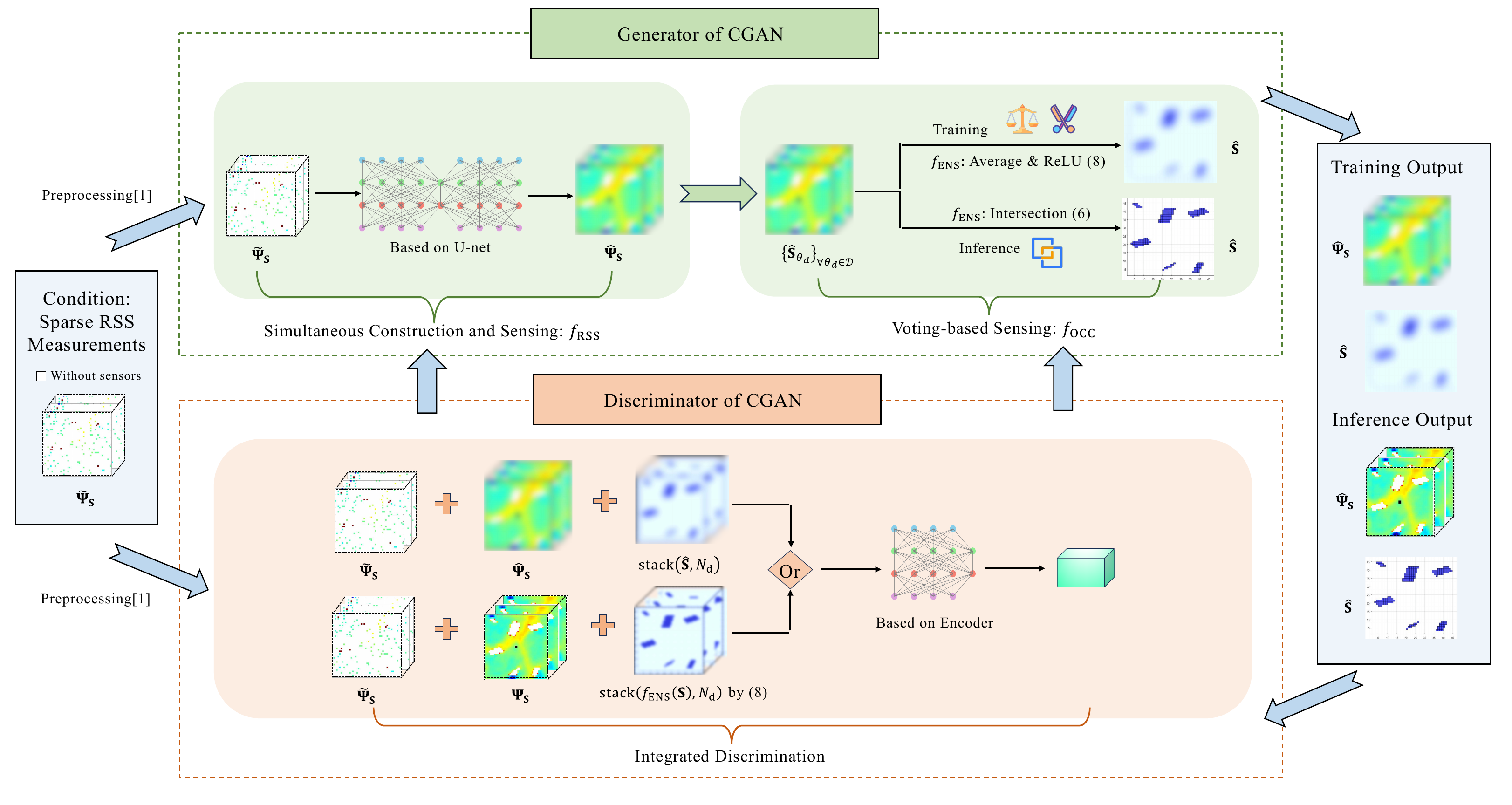}
\vspace{-20pt} 
\caption{\centering The proposed CGAN architecture for integrated THz radio map construction and obstacle sensing.}
\label{overall_framework}
\end{figure*}

\section{CGAN-based Framework for Integration}
\label{sec:prop}  
To address the complex problem outlined in~\eqref{final_problem}, we leverage a deep learning model based on \acp{cgan}, utilizing a finite set of environmental scenarios derived from historical measurements. Specifically, we propose a \ac{cgan}-based framework for solving the optimization problem in \eqref{final_problem_obj}. This approach involves offline adversarial training combined with online inference, where "soft" and "hard" voting strategies are employed to achieve a well-performing model for integrated tasks, such as directional \ac{thz} radio map construction and obstacle sensing. It is important to highlight that the proposed \acp{cgan} framework is adept at capturing the complex relationships between obstacle layouts, sparse \ac{rss} measurements, and directional \ac{thz} radio maps. This capability arises from \ac{cgan}'s strong generalization power and its integration approach, which extends beyond conventional techniques. In the context of generative models, \acp{gan} strikes a balance between accuracy, efficiency, and stability due to its conditional nature, making it a robust choice for learning the \ac{thz} signal propagation mechanism in relation to varying obstacle distributions. Moreover, the adversarial training process of \ac{cgan} reflects the synergistic enhancement between map construction and obstacle sensing, offering a solution to the challenge of balancing regression and classification tasks. Consequently, the proposed \ac{cgan}-based framework demonstrates excellent generalization performance across diverse target environmental scenarios.


\subsection{The Proposed CGAN Architecture}
The \ac{cgan} structure~\cite{isola2017image} can be utilized to address the objective \eqref{final_problem_obj} by performing the two integrated tasks. In line with \ac{cgan}, the proposed architecture consists of two key sub-networks: the generator \( G \) and the discriminator \( D \), as illustrated in Fig.~\ref{overall_framework}. To mitigate the computational complexity of our architecture and relax the framework's dependency on the number \( N_{\dR} \) of beams during inference, we interpret the beam dimension of \( \PsiB_{\SF} \) as the batch dimension rather than as a structural component. 
For a given environmental scenario, the generator \( G \) treats the \ac{rss} measurements (following the same preprocessing steps and predefined threshold \( \Psi_{\smax} \) for obstacle segmentation as in \cite{hu2024towards}) as the conditional input. It then generates \( N_{\dR} \) directional \ac{thz} radio maps and obstacle sensing results during training or inference, i.e., \( \{\wh{\PsiB}_{\SF}, \hat{\SF}\} = G(\wt{\PsiB}_{\SF}; \mathbf{\Theta}_{g}) \), where \( \mathbf{\Theta}_{g} \) represents the model parameters. The discriminator \( D \) functions as a classifier, assessing the plausibility of the outputs generated by \( G \) in relation to the real directional \ac{thz} radio maps and the obstacle layout. The two sub-networks are trained in an adversarial manner to solve the problem outlined in \eqref{final_problem}. This adversarial training not only ensures that the outputs of the generator \( G \) are realistic, but also enhances the generalization capability of \( G \) (see Sec.~\ref{sec:offline} for details). In addition to the adversarial training, the structures of both \( G \) and \( D \) are carefully designed to further improve generalization, as discussed below.

\subsubsection{CGAN Generator}
The generator \( G \) is composed of two modules: the simultaneous construction and sensing module, which implements the function \( f_{\RSS} \), and the voting-based sensing module, which performs the function \( f_{\ENS} \). The simultaneous construction and sensing module employs an enhanced encoder-decoder structure, specifically U-Net~\cite{ronneberger2015u}, to effectively capture the low-dimensional manifold nature of directional \ac{thz} radio maps. Furthermore, according to~\eqref{occupancy_probability}, this module simultaneously enables the segmentation of obstacles from the constructed radio maps \( \{[\wh{\PsiB}_{\SF}]_d\}_{d=1}^{N_{\dR}} \) in each direction, although obstacle segmentation is not performed during the integration phase. As shown in Fig.~\ref{network}, the architecture of this module consists of an approximately symmetric encoder-decoder structure with skip connections. The encoder compresses each element of the batch, i.e., the conditional input \( [\wt{\PsiB}_{\SF}]_d \), into a low-dimensional vector, while the decoder transforms this vector into the estimated radio map \( [\wh{\PsiB}_{\SF}]_d \).

It is important to note that the encoder-decoder structure is well-suited for directional \ac{thz} radio map construction because the radio maps are distributed on a low-dimensional manifold embedded within a high-dimensional space \cite{9523765}. This manifold can be represented by the aforementioned low-dimensional vector, which enhances the generalization ability of the radio map construction module when properly trained on target environmental scenarios. In addition, the skip connection in the U-Net enables the network to accommodate various receptive field sizes, facilitating the sensing of obstacles of different sizes. Since \acp{cnn} exploit the spatial local correlations of radio maps and obstacle layouts with translational invariance, the U-Net structure is implemented using 2-dimensional (2D) \acp{cnn}, as shown in Fig.~\ref{network}. Furthermore, the sigmoid activation function, \( \sigmoid(x) = (1 + e^{-x})^{-1} \), is applied after the final convolutional layer to ensure that \( [\wh{\PsiB}_{\SF}]_{i,j,d} \in [0, 1] \).

The voting-based sensing module is designed to implement both the "soft" and "hard" voting strategies, as defined in \eqref{soft_vote} and \eqref{hard_vote}, to generate more accurate obstacle sensing results during training and inference, respectively. In particular, the "soft" voting strategy is utilized to facilitate gradient-based optimization during \ac{cgan} training, where \( \max(0,x) \) in \eqref{soft_vote} is achieved using the ReLU activation function. In contrast, during inference, where gradients are not required, the "hard" voting strategy is employed, which does not compromise the performance of the system. As depicted in the proposed architecture shown in Fig. \ref{overall_framework}, this module explicitly features separate branches for the different voting strategies. The generator's outputs are accordingly induced, providing both the estimated radio maps \( \wh{\PsiB}_{\SF} \) and the associated occupancy confidence levels during training, as well as discrete estimations during inference for all grid points. Furthermore, the input of this module corresponds to the output of the simultaneous construction and sensing module. Specifically, we set \( \{\hat{\SF}_{\theta_{d}}\}_{\theta_{d}\in\DC} = \{[\wh{\PsiB}_{\SF}]_d\}_{d=1}^{N_{\dR}} \) based on the preprocessing method and subsequent segmentation approach described in~\cite{hu2024towards}, as discussed in Section~\ref{subsec:vote}. In summary, this module performs environmental semantics extraction via the voting strategy \( f_{\ENS} \) applied to the set \( \{\hat{\SF}_{\theta_{d}}\}_{\theta_{d}\in\DC} \), which is separately sensed by the U-Net module.


\subsubsection{CGAN Discriminator}
The discriminator \( D \) with parameters \( \mathbf{\Theta}_d \) aims to discriminate whether its input consists of synthetic data from the generator or actual data from the available historical measurements during offline training. Furthermore, the discrimination results are utilized to ensure the generalization capability and assist in improving the generation capability. Therefore, considering our integrated tasks for \ac{thz} radio map construction and obstacle sensing, the discriminator of \ac{cgan} is designed as an integration discrimination module in Fig.~\ref{overall_framework}, where the input data includes the conditional term \( \wt{\PsiB}_{\SF} \), the regression terms \( \wh{\PsiB}_{\SF} \) versus \( \PsiB_{\SF} \), and the classification terms \( \hat{\SF} \) versus \( \SF \). That is, with the assistance of \( \wt{\PsiB}_{\SF} \), the discrimination of \ac{thz} radio maps and obstacles is conducted. Specifically, we set \( D(\wh{\PsiB}_{\SF}, \stack(\hat{\SF}, N_{\dR}) \mid \wt{\PsiB}_{\SF}; \mathbf{\Theta}_d) \) and \( D(\PsiB_{\SF}, \stack(f_{\ENS}(\SF), N_{\dR}) \mid \wt{\PsiB}_{\SF}; \mathbf{\Theta}_d) \) for the corresponding discrimination, where \( f_{\ENS}(\SF) = f_{\ENS}(\{\SF\}_{\theta_d \in \DC}) \) is calculated by \eqref{soft_vote} to contrast with \( \hat{\SF} \), while \( \stack(\hat{\SF}, N_{\dR}) \) and \( \stack(f_{\ENS}(\SF), N_{\dR}) \) are used to keep the dimension consistent with multi-directional \ac{thz} radio maps \( \PsiB_{\SF} \). For notational simplicity, the outputs of \( D \) are designated as \( D(\wh{\PsiB}_{\SF}, \hat{\SF} \mid \wt{\PsiB}_{\SF}; \mathbf{\Theta}_d) \) and \( D(\PsiB_{\SF}, f_{\ENS}(\SF) \mid \wt{\PsiB}_{\SF}; \mathbf{\Theta}_d) \). In addition, as depicted in Fig.~\ref{network}, the discriminator structure exploits the encoder part of the module \( f_{\RSS} \).

\begin{figure}[!t]
\center
\includegraphics[width=8.5cm,height=6.5cm] {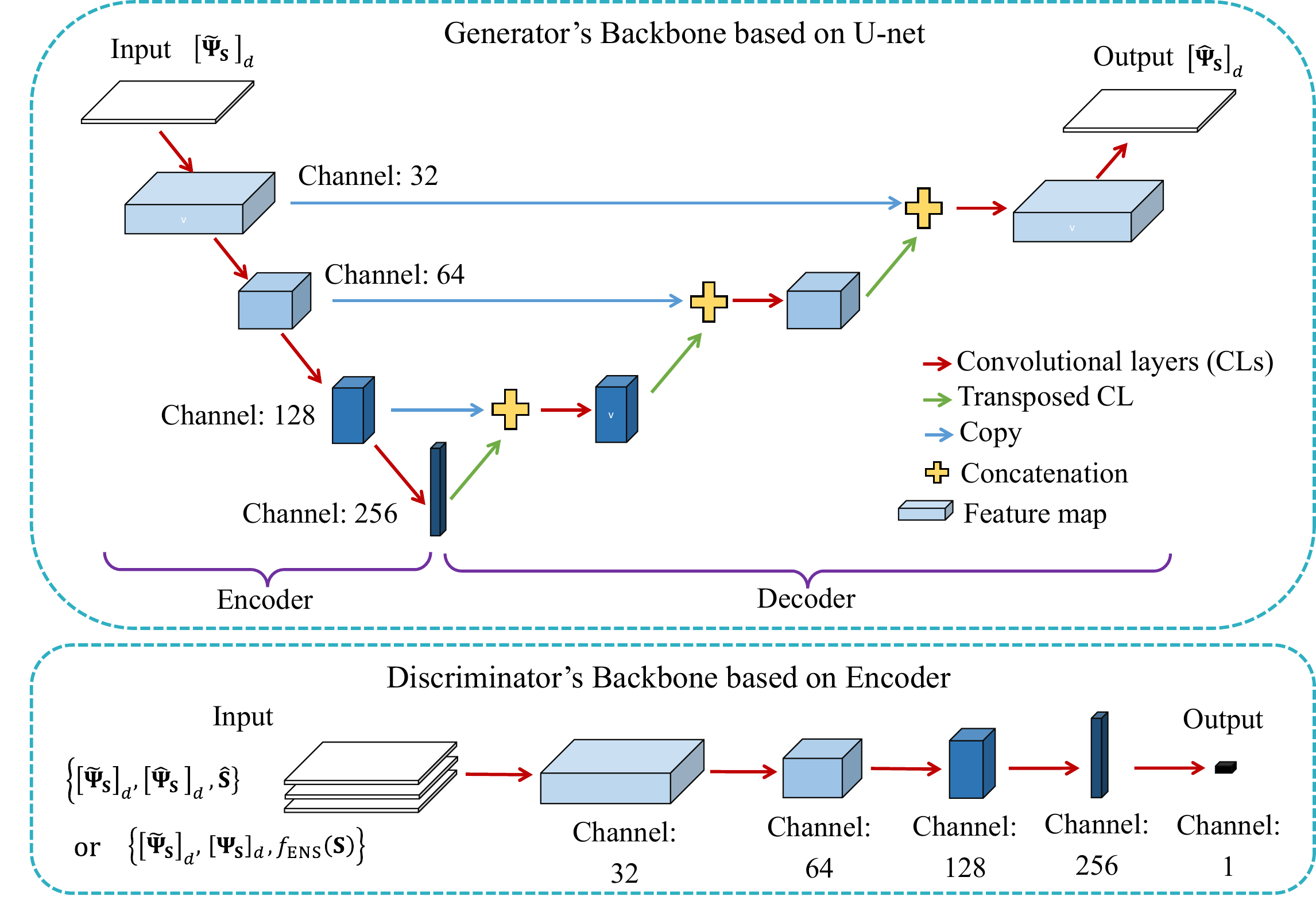}
\caption{The network structure of the proposed CGAN architecture.}
\label{network}
\end{figure}

\subsection{Off-Line Adversarial Training}
\label{sec:offline}

\addtolength{\topmargin}{0.03in}
 
To obtain an approximate solution to problem \eqref{final_problem}, one need to train the generator $G$ and discriminator $D$ in the proposed CGAN architecture with a finite set of historical measurements, where $N_{\tR}$ samples with the form $\{\PsiB_{\SF}^{t},\SF^{t}\}$ are available. To this end, problem \eqref{final_problem} is approximatedly reformulated as the following empirical version with  $\mathbf{\Theta}_{g}$ as the optimization variable:
\begin{equation} \label{trainining_problem}
\begin{aligned}
    \underset{\mathbf{\Theta}_{g}}{\mathrm{minimize}}\,\,V(G) &=\frac{1}{N_{\tR}}\sum_{t=1}^{N_{\tR}}\left(B(\PsiB_{\SF}^{t},\wh{\PsiB}_{\SF}^{t})+\eta C(\SF^{t}, \hat{\SF}^{t})\right), 
\end{aligned} 
\end{equation} where $\{\wh{\PsiB}_{\SF}^{t},\hat{\SF}^{t}\}=G(\wt{\PsiB}_{\SF}^{t};\mathbf{\Theta}_{g})$, and other constraints are absorbed into \eqref{trainining_problem}. Note that $\wt{\PsiB}_{\SF}^{t}$ is obtained by randomly sampling $K\times N_{\dR}$ RSS values from $\PsiB_{\SF}^{t}$, where $\frac{K}{N_{\WR}N_{\LR}}$ is the sampling rate for each THz radio map. A straightforward approach to solving problem \eqref{trainining_problem} is to train the generator $G$ alone with the same training objective. However, this approach may lead to deficient generalization performance, due to the lack of adversarial guidance, which prevents $G$ from capturing the distributions of $\PsiB_{\SF}$ and $\SF$.

The training procedure of $G$ and $D$ in CGAN follows an iterative process, where the two networks are optimized alternately and adversarially to each other based on problem \eqref{trainining_problem}, resembling a two-player zero-sum minimax game. The idea is that the generator $G$ is trained to generate the samples closer to actual data to fool $D$, and the discriminator $D$ is trained to accurately discriminate the generated and actual data so as not to be fooled. Specifically, at each iteration of the training process, $G$ and $D$ are updated by alternatively approximating the following sub-problems with multi-step stochastic gradient descent:
\begin{subequations}
\begin{align}
    &\underset{\mathbf{\Theta}_{g}}{\mathrm{minimize}}\,\,\, \alpha V(G)-\beta_{1} V_{1}(G,D),
    \label{trainining_loss_G} \\
    &\underset{\mathbf{\Theta}_{d}}{\mathrm{maximize}}\,\,\, V_{2}(G,D)- V_{1}(G,D)-\beta_{\PR} V_{\PR}(G,D),
    \label{trainining_loss_D}
\end{align} 
\end{subequations}
where $\alpha$, $\beta_{1}$, and $\beta_{\PR}$ are the weights of different term. Note that the terms $V_{1}(G, D)$ and $V_{2}(G, D)$ are derived from CGAN for generalization enhancement, and are given as 
\begin{align}
    V_{1}(G,D)=\frac{1}{N_{\tR}}\sum_{t=1}^{N_{\tR}}D(\wh{\PsiB}_{\SF}^{t},\hat{\SF}^{t}\vert\wt{\PsiB}_{\SF}^{t}; \mathbf{\Theta}_{d})
    \label{trainining_loss_G_sub},
\end{align} 
\begin{align}
    V_{2}(G,D)=\frac{1}{N_{\tR}}\sum_{t=1}^{N_{\tR}}D(\PsiB_{\SF}^{t},f_{\ENS}(\SF^{t})\vert\wt{\PsiB}_{\SF}^{t}; \mathbf{\Theta}_{d}).
    \label{trainining_loss_D_sub}
\end{align} 

Moreover, the penalty term $V_{\PR}(G,D)$ with $\beta_{\PR}$ is introduced to improve the training stability \cite{gulrajani2017improved} of CGAN, and 
\begin{align}
V_{\PR}(G,D)=&\frac{1}{N_{\tR}}\sum_{t=1}^{N_{\tR}}\big[(\|\nabla_{\Upsilon}D(\Upsilon\vert\wt{\PsiB}_{\SF}^{t}; \mathbf{\Theta}_{d})\|_{2}-1)^{2} ,
\label{trainining_loss_D_pan}
\end{align}where $\Upsilon=\{\wh{\PsiB}_{\SF}^{t'},\hat{\SF}^{t'}\}$, $\wh{\PsiB}_{\SF}^{t'}=\gamma \wh{\PsiB}_{\SF}^{t}+(1-\gamma)\PsiB_{\SF}^{t}$ and $\hat{\SF}^{t'}=\gamma \hat{\SF}^{t}+(1-\gamma)f_{\ENS}(\SF^{t})$. Note that $\gamma$ is a random variable that follows a uniform distribution over the interval $[0,1]$.

As observed from \eqref{trainining_loss_G_sub} and \eqref{trainining_loss_D_sub}, the term \( V_1(G, D) \) evaluates the performance of the generator by assessing the discrimination between the generated radio maps and sensed obstacles, while the term \( V_2(G, D) \) indicates the direction in which the generator should progress, based on the discrimination between the true maps and obstacles. Comparing \eqref{trainining_loss_G} and \eqref{trainining_loss_D}, the generator seeks to minimize the difference between \( V_1(G, D) \) and \( V_2(G, D) \) through training, whereas the discriminator aims to maximize the difference between these terms, reflecting the adversarial nature of the training objectives. Thus, following the principle of alternative training, the discriminator \( D \) serves as a metric for assessing the progress of \( G \), providing feedback on how well the generated radio maps and obstacle sensing results align with the ground truth. This dynamic drives the competitive interaction that leads to improved generation and generalization performance.

The detailed training procedure, summarized in Algorithm~\ref{algo:gantrain}, is similar to that in~\cite{goodfellow2014generative}. It involves multiple training epochs, where the generator \( G \) is trained using multi-step stochastic gradient descent to approximate \eqref{trainining_loss_G}, with \( D \) fixed. This is followed by the training of the discriminator \( D \), which is optimized to approximate \eqref{trainining_loss_D} with \( G \) fixed, until an equilibrium point is reached. It is worth noting that the proposed training objectives, including the term \( V(G) \) and the \ac{gan}-related terms, effectively guide the generator \( G \) toward solving the problem in \eqref{final_problem} with desirable generalization performance. This results in the generation of precise \ac{thz} directional radio maps and accurate obstacle sensing for the target environmental scenarios.

\begin{algorithm}
\caption{Off-Line Adversarial Training With CGAN-based Framework for Integration}
\label{algo:gantrain}
\begin{algorithmic}[1]%
\STATE \textbf{Initialize:} $\mathbf{\Theta}_{g}$, $\mathbf{\Theta}_{d}$, $\alpha$, $\beta_{1}$, $\beta_{\PR}$, and a number $\zeta=0$. %
\STATE \textbf{for} number of training epochs \textbf{do}
\STATE \hspace{0.5em} \textbf{for} ${N_{\tR}}$ \textbf{do}
\STATE \hspace{1.2em} Sample $\PsiB_{\SF}^{t}$, $\wt{\PsiB}_{\SF}^{t}$, and $\SF^{t}$ from training data;%
\STATE \hspace{1.2em} \textbf{if} $\zeta==0$ \textbf{then}
\STATE \hspace{2.1em} Update $D$ by ascending its stochastic gradient de-
\STATE \hspace{2.1em} rived from the minibatch loss aligned with \eqref{trainining_loss_D};%
\STATE \hspace{2.1em} Set $\zeta=1$;%
\STATE \hspace{1.2em} \textbf{else} 
\STATE \hspace{2.1em} Update $G$ by descending its stochastic gradient de-
\STATE \hspace{2.1em} rived from the minibatch loss aligned with \eqref{trainining_loss_G};%
\STATE \hspace{2.1em} Set $\zeta=0$.%
\STATE \textbf{Output:} The trained generator $G$.%
\end{algorithmic}
\end{algorithm}

\subsection{Online Inference}
The trained generator \( G \) in the proposed \ac{cgan} framework can be applied to any target environmental scenario \( \SF \), using a limited set of \ac{rss} measurements \( \wt{\PsiB}_{\SF} \) as input. Consequently, the outputs of \( G \) are the corresponding scaled \ac{thz} radio maps \( \wh{\PsiB}_{\SF} \) and the sensed obstacles \( \hat{\SF} \), obtained through the "hard" voting strategy in \eqref{hard_vote}. By applying the inverse operation of the preprocessing step, the \ac{thz} radio maps with actual \ac{rss} values can be recovered.

\begin{figure}[!t]
\centering
\includegraphics[height=4cm] {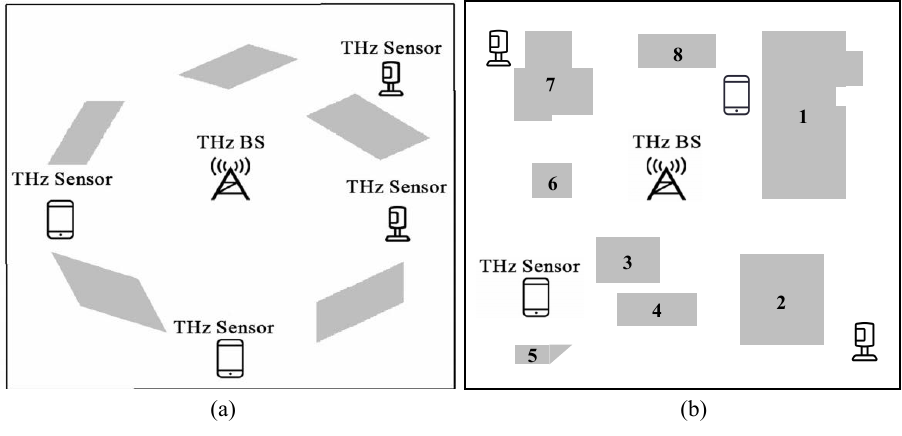}
\caption{(a) An example obstacle layout of the simulated environment; (b) The obstacle layout form the section of Oklahoma of Fig. \ref{new_scenario}(b), i.e. real-world city.}
\label{evaluation_scenario}
\end{figure}

\renewcommand{\arraystretch}{1.5} 
\begin{table}[htbp]
\caption{Scenarios details and model' hyper-parameters.}
\label{tab_evalution}
\begin{supertabular}{>{\centering\arraybackslash}m{0.18\textwidth}  
|>{\centering\arraybackslash}m{0.26\textwidth}} \toprule 
\textbf{Description of Parameter} & \textbf{Value} \\ \midrule
Scenario size & $L\times W=100\mR\times100\mR$  \\ \hline
Grid division & $N_{\LR}\times N_{\WR}=64 \times 64$  \\ \midrule
Number of obstacles & $\SC_1$: $\{1,2,3,4,5\}$; $\SC_2$: $\{6\}$; $\SC_3$: $\{8\}$ \\ \midrule
Number $N'$ of $\{\SF,\PsiB_{\SF}\}$ & $\SC_1$: $250$; $\SC_2$: $50$; $\SC_3$: $1$ \\ \midrule
Carrier frequency & $300\,\mathrm{GHz}$ \\ \midrule
Number of beams & $N_{\dR}=18$ \\ \midrule
Beamwidth & $\theta_{\bR}=20^{\circ}$ \\ \hline 
Angular separation & $\Delta\theta=20^{\circ}$ \\ \midrule
Sampling rate & $0.5$ \\ \midrule
Weights in \eqref{trainining_loss_G} and \eqref{trainining_loss_D} & $\alpha=1000$; $\beta_{1}=10$; $\beta_{\PR}=10$ \\ \midrule
Number of training epochs & $100$ \\ \midrule 
\multirow{2}{*}{Indicator $\eta$} & $\eta=1$ if \acs{cgan} has trained for at least\\  
 & $50$ epochs and $\eta=0$ otherwise\\\hline
 Learning rate & $0.0001$ \\ \bottomrule
\end{supertabular}
\end{table}

\section{Experimental Results}
\label{sec:res}

\subsection{Experimental Setup}



\begin{figure*}[t]
    \centering
    \subfloat[MSE of THz radio
map construction task]{%
        \includegraphics[width=0.33\textwidth]{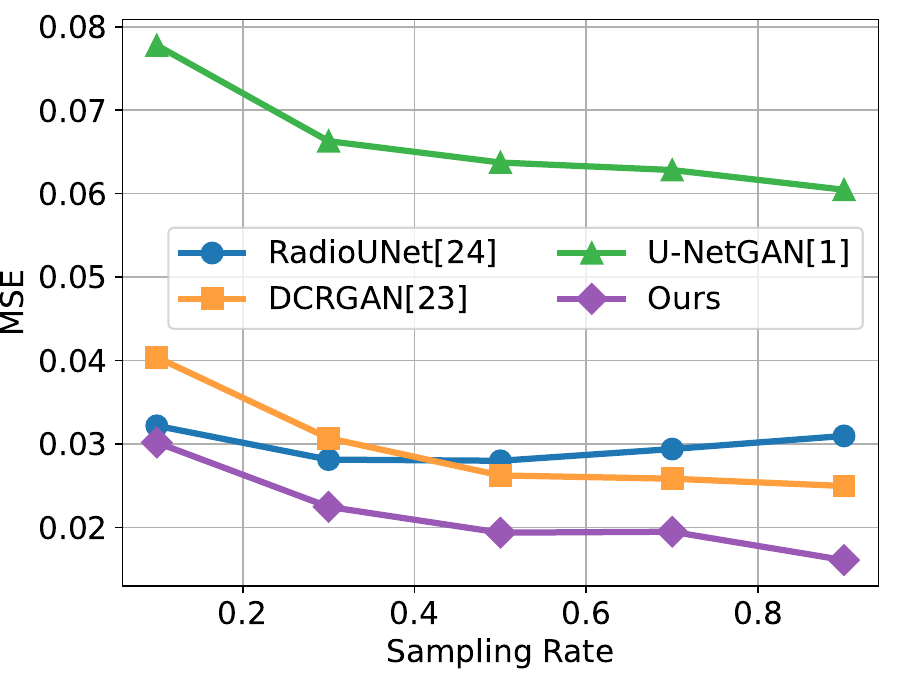}
        \label{fig:integration_evaCons}
    } 
    \subfloat[MSE of obstacle sensing task]{%
\includegraphics[width=0.33\textwidth]{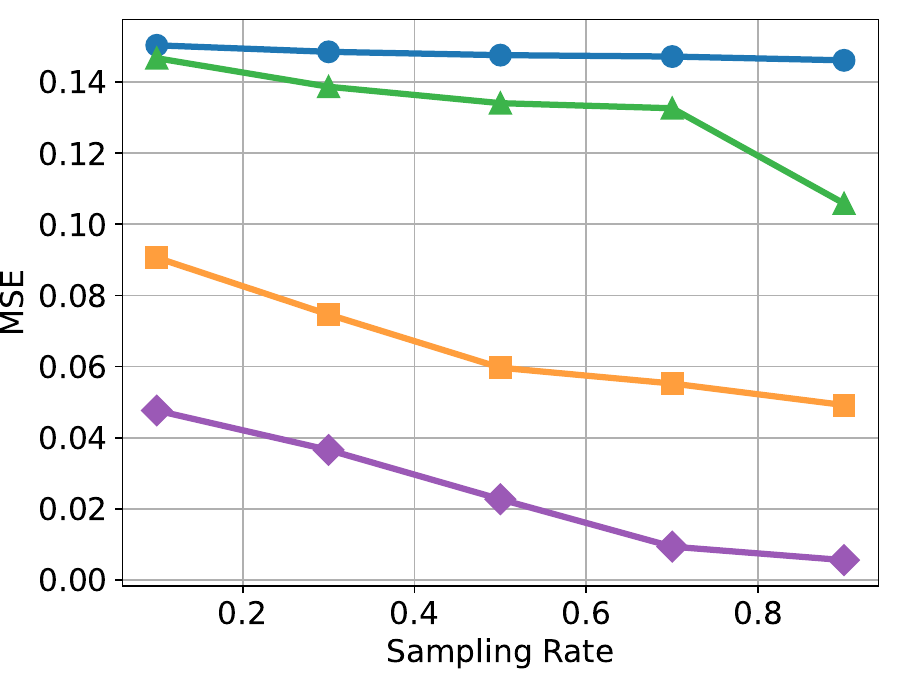}
        \label{fig:integration_evaobst}
    }
        \subfloat[AP of obstacle sensing task]{%
\includegraphics[width=0.33\textwidth]{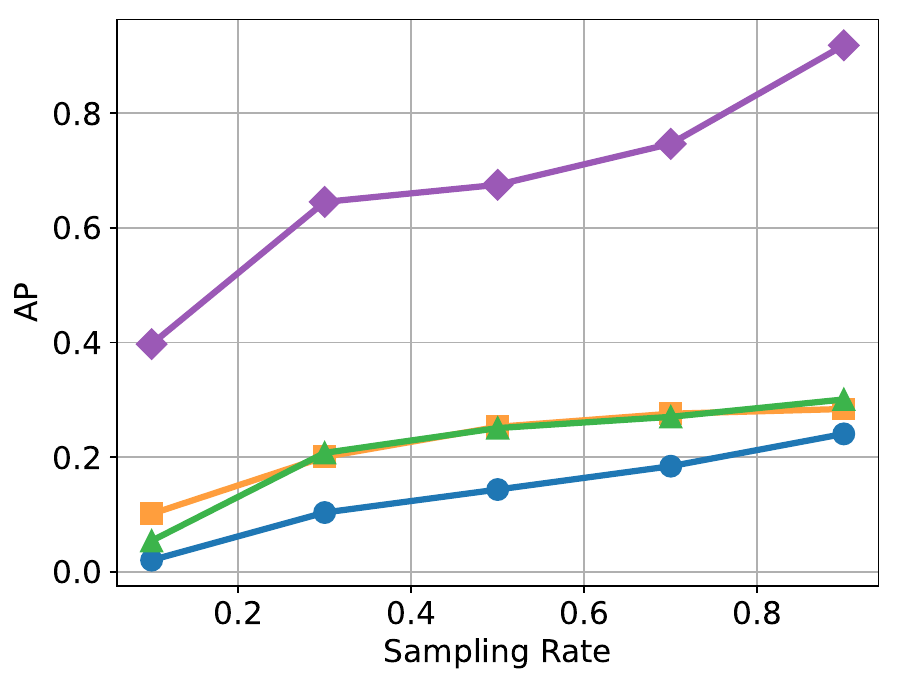}
        \label{fig:integration_evaAP}
    }
    \caption{Performance of different tasks versus sampling rate for different approaches evaluated in $\SC_2$.}
    \label{fig:integration_eva}
\end{figure*}


{\bf Datasets and Model Hyperparameters.} We consider three types of 2D scenarios, which include training and validation scenarios with a simulated environment, as well as a validation scenario based on a real-world city. These scenarios are labeled as \( \SC_1 \), \( \SC_2 \), and \( \SC_3 \), respectively. Fig.~\ref{evaluation_scenario}(a) presents an example scenario of a simulated environment, where obstacles have random quadrilateral shapes and positions, and the number of obstacles is also randomly determined. Fig.~\ref{evaluation_scenario}(b) illustrates the obstacle layout corresponding to Fig.~\ref{new_scenario}(b). For each scenario in \( \SC_1 \), \( \SC_2 \), and \( \SC_3 \), the corresponding \ac{thz} radio maps are generated through ray tracing~\cite{remcom}, which models the power attenuation function \( \Gamma_{\SF}(\xB_{i,j}, \theta_{d}, \theta_{\bR}) \) implicitly. Each ray is allowed to undergo at most one reflection or diffraction, as the \ac{thz} signal is subject to molecular absorption and significant path loss. Unless otherwise stated, the default parameters for ray tracing and the scenarios are provided in Table~\ref{tab_evalution}, which also includes the hyperparameters of the proposed \ac{cgan}-based framework. Based on the information in Table~\ref{tab_evalution}, the data (e.g., \( \{\PsiB_{\SF}^{t}, \SF^{t}\} \)) corresponding to \( \SC_1 \) is used for offline training, while data for \( \SC_2 \) and \( \SC_3 \) are used for online inference, which are employed to assess performance and generalization. \\\\
{\bf Baseline Models and Evaluation Metrics.} Existing approaches do not directly address the integration of \ac{thz} radio map construction and obstacle sensing. Therefore, we use U-Net-based and \ac{gan}-based radio map construction approaches~\cite{hu2024towards, hu20233d, levie2021radiounet} as baselines to simultaneously perform both tasks. These are denoted as U-NetGAN, DCRGAN, and RadioUNet, respectively. The obstacle sensing idea from~\cite{hu2024towards} is applied in these baseline models. Note that the method from \cite{10626572} is not included in the comparison, as it focuses on a federated learning model and environment-specific training and deployment, similar to \cite{8269065, 9500337, 9154223, 10041012, zeng2022uav, 10437033}. Once the proposed CGAN-based framework and the baselines are well trained on scenarios \( \SC_1 \), their performance is evaluated on  scenarios \( \SC_2 \) and \( \SC_3 \).

For evaluation metrics, we use the \ac{mse} with respect to both construction and sensing, specifically \( \frac{1}{N'} \sum \|\SF - \hat{\SF}\|_{\FR}^{2} \) and \( \frac{1}{N'} \sum \|\PsiB_{\SF} - \wh{\PsiB}_{\SF}\|_{\FR}^{2} \). Additionally, we also adopt the AP, originally designed for object detection \cite{padilla2020survey}, as a sensing metric, since AP used in \cite{hu2024towards} can be adapted to evaluate the overlap between the sensed and true obstacles from both recall and precision perspectives. In our simulations, AP is computed as the average of the approximate areas under the Precision-Recall (PR) curves for various intersection over union (IoU) thresholds ranging from $0.5$ to $0.95$ with a step of $0.05$, where each PR curve is obtained by plotting precision against recall at a given IOU threshold. As a result, a higher AP value indicates that obstacles are detected with minimal omissions and that non-obstacle regions are rarely misclassified as obstacles, making it a more suitable metric for assessing sensing performance and addressing the limitation of MSE, which focuses solely on pixel-wise accuracy and may provide an incomplete evaluation. Note that the cross-entropy in \eqref{trainining_problem} cannot be used as metric, as it only measures distribution differences. To account for the randomness in neural network initialization, each experiment is repeated five times, and the reported results correspond to the average performance over these runs.

\renewcommand{\arraystretch}{1.25} 
\begin{table*}[ht]
\centering
\caption{Impact of Different Integration Aspects\\ (mean $\pm$ standard deviation of MSE and AP over the last 10 epochs evaluated in $\SC_{2}$)}
\label{tab_integration}
\begin{supertabular}{>{\centering\arraybackslash}m{0.15\textwidth}|>{\centering\arraybackslash}m{0.25\textwidth} 
|>{\centering\arraybackslash}m{0.16\textwidth}|>{\centering\arraybackslash}m{0.16\textwidth}|>{\centering\arraybackslash}m{0.16\textwidth}} \toprule 
\multirow{2}{*}{\textbf{Integration Aspect}} &
\multirow{2}{*}{\textbf{Description}} &
\textbf{Construction for $\PsiB_{\SF}$}& 
\textbf{Sensing for $\SF$} &
\textbf{Sensing for $\SF$} \\
&&
\textbf{(MSE)}& 
\textbf{(MSE)} &
\textbf{(AP)} \\ \midrule
Proposed & With all integration aspects  & $\mathbf{0.0179\pm0.0008}$
 &  $\mathbf{0.0187\pm0.0025}$ & $\mathbf{0.7378\pm0.0330}$\\ \hline
\multirow{2}{*}{Training loss} & Without $B(\PsiB_{\SF},\wh{\PsiB}_{\SF})$ & $0.0807\pm0.0040$ & $0.1416\pm 0.0029$ & $0.0036\pm 0.0020$\\ 
& Without $C(\SF,\hat{\SF})$ & $ 0.0485\pm 0.0002$ & $0.1912\pm 0.0628$ & $0.1104\pm 0.0201$ \\ \midrule
\multirow{4}{*}{Discriminator} 
& Without $\wh{\PsiB}_{\SF}$ and $\PsiB_{\SF}$ & $0.0192\pm 0.0005$ & $0.0148\pm0.0016$ & $0.7991\pm 0.0210$\\
& Without $\hat{\SF}$  and $f_\ENS(\SF)$ & $0.0192\pm0.0013$ & $0.0214\pm 0.0042$ & $0.7024\pm 0.0411$\\
& $\wh{\PsiB}_{\SF}$ and  $\PsiB_{\SF}$ within $[\Psi_{\smin},\Psi_{\smax}]$ & $0.0177\pm 0.0014$ & $0.0215\pm0.0058$ & $0.6988\pm 0.0558$ \\
& $\wh{\PsiB}_{\SF}$ and  $\PsiB_{\SF}$ within $[\Psi_{\smax},1]$  & $0.0185\pm0.0012$ & $0.0183\pm0.0060$ &  $0.7592\pm 0.0463$\\ \midrule
Preprocessing & Without Preprocessing  & $0.0514\pm 0$ & $0.6575\pm0$ & $0\pm 0$\\ \midrule
Weight $\PsiB_{\SF}$ & Without weight $\PsiB_{\SF}$ & $0.0274\pm 0.0031$ & $0.0587\pm0.0105$ & $0.4066\pm 0.0652$\\  
\bottomrule
\end{supertabular}
\end{table*}


\begin{figure*}[t]
    \centering
    \subfloat[MSE of THz radio
map construction task]{%
        \includegraphics[width=0.33\textwidth]{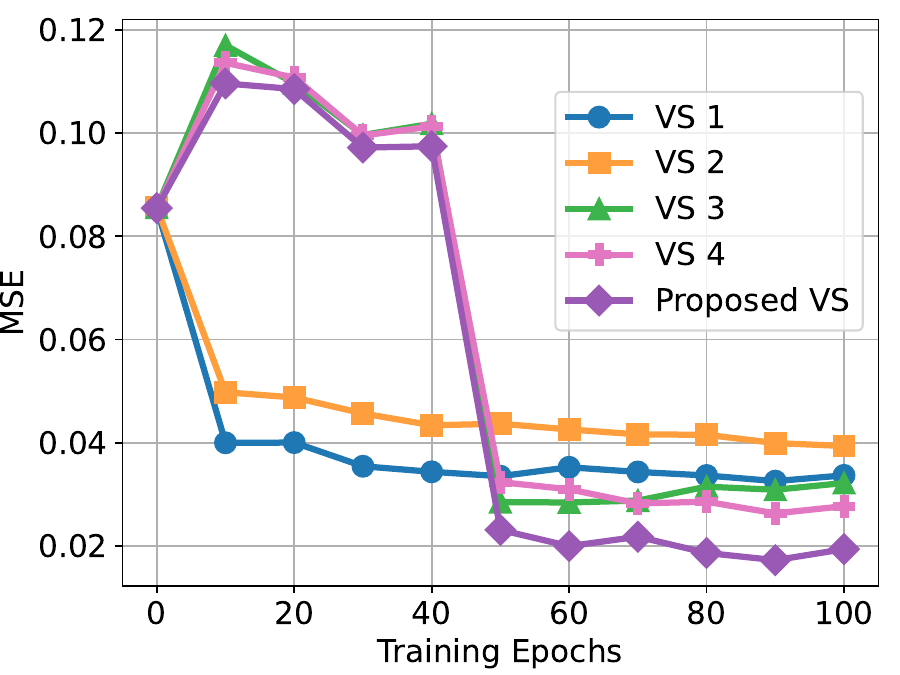}
        \label{fig:voting_evaCons}
    } 
    \subfloat[MSE of obstacle sensing task]{%
\includegraphics[width=0.33\textwidth]{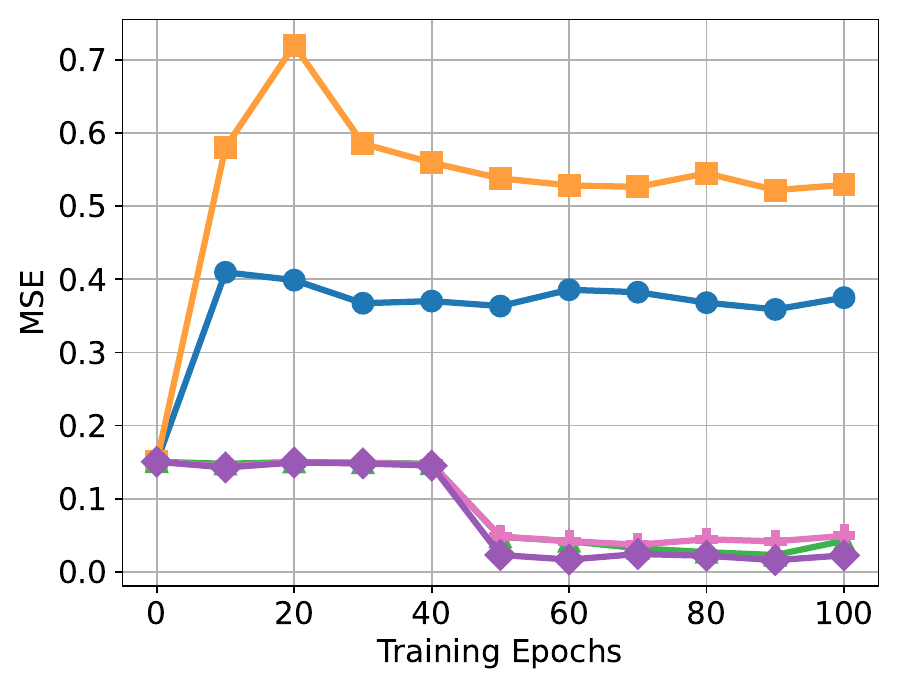}
        \label{fig:voting_evaobst}
    }
        \subfloat[AP of obstacle sensing task]{%
\includegraphics[width=0.33\textwidth]{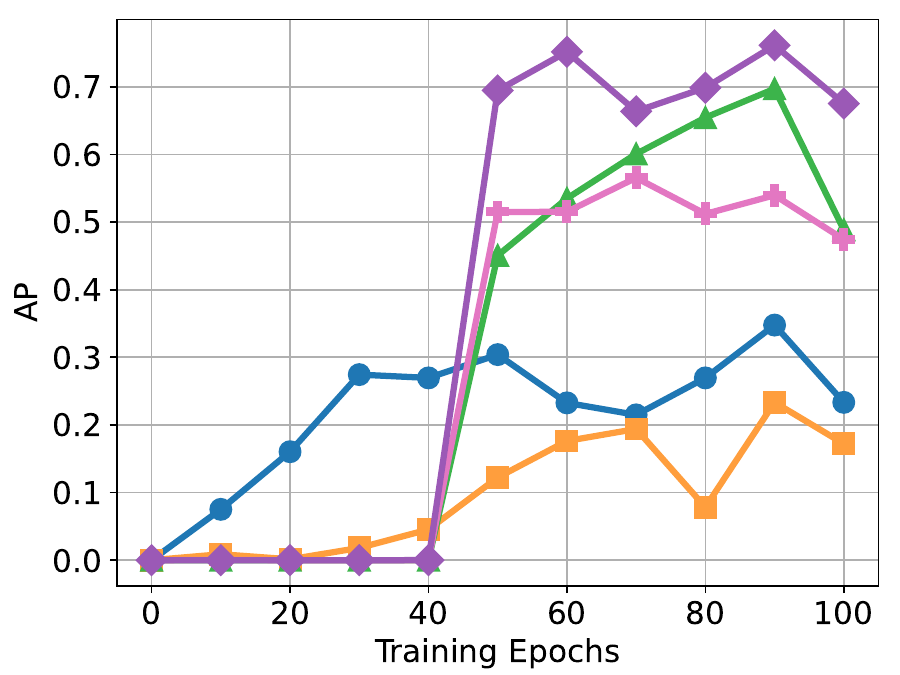}
        \label{fig:voting_evaAP}
    }
    \caption{Performance of different tasks versus training epochs for different voting strategies evaluated in $\SC_2$. VS: Voting Strategy. }
    \label{fig:voting_eva}
\end{figure*}
\subsection{Effectiveness of Integration}
To evaluate the performance of the proposed integration framework against the baseline models, Fig. \ref{fig:integration_evaCons} and \ref{fig:integration_evaobst} present the \ac{mse} with respect to \ac{thz} radio map construction and obstacle sensing as a function of the sampling rate, while Fig.~\ref{fig:integration_evaAP} illustrates the AP  to further evaluate sensing performance under the same configurations. As can be seen from Fig.~\ref{fig:integration_eva}, the MSE in both tasks generally decreases with an increase in sampling rate, whereas AP exhibits the opposite trend, increasing as the sampling rate rises, since a higher AP value reflects greater overlap and thus better sensing of obstacles’ location and shape. Notably, the proposed framework consistently outperforms the non-integrated baselines under scenario \( \SC_2 \), which differs from the training scenario, not only in terms of MSE but also in AP. This result underscores the effectiveness of the integration approach and its corresponding generalization capability. Furthermore, the observation that U-NetGAN~\cite{hu2024towards} and DCRGAN~\cite{hu20233d}, which exhibit significantly different MSE values, achieve nearly identical AP scores highlights that MSE and AP each capture distinct aspects of the overall sensing performance. In addition, the \ac{mse} achieved by RadioUNet~\cite{levie2021radiounet} does not decrease with increasing sampling rates. This may be attributed to the fact that RadioUNet, which does not implement integration, leverages additional prior physical environmental maps that are not available in our integrated problem.

Table~\ref{tab_integration} compares the impact of various integration aspects on performance, based on the mean and standard deviation of \ac{mse} and AP during the last $10$ convergence epochs. In this table, the term "training loss" refers to the consideration of construction and sensing objectives during training, while the "discriminator" term evaluates the effect of different input components related to the generator \( G \) on performance in both tasks. In addition, the influences of preprocessing and the weight \( \PsiB_{\SF} \) are also assessed. The results indicate that the proposed framework achieves a well-balanced overall performance in terms of both the mean and the deviation. Some integration aspects may perform better in one task but underperform in another. This balance provides flexibility in adjusting the integration approach based on task-specific priorities during deployment. For example, when RSS information (i.e., $\wh{\PsiB}_{\SF}$ and $\PsiB_{\SF}$) is excluded from the discriminator, both the construction performance and the obstacle sensing performance show smaller standard deviations, suggesting improved stability. Moreover, the sensing capability, as captured by the mean values of AP and MSE, improves simultaneously. However, the construction capability reflected by mean MSE deteriorates. Across various integration aspects, the trends in AP and MSE remain notably similar, despite their reflection of different facets of sensing performance.

\subsection{Impact of Voting}

\renewcommand{\arraystretch}{1.25} 
\begin{table*}[ht]
\centering
\caption{Impact of Beam Characteristics for Voting  \\(mean $\pm$ standard deviation of MSE and AP over the last 10 epochs evaluated in $\SC_{2}$).}
\label{tab_voting}
\begin{supertabular}{>{\centering\arraybackslash}m{0.15\textwidth}|>{\centering\arraybackslash}m{0.25\textwidth} 
|>{\centering\arraybackslash}m{0.16\textwidth}|>{\centering\arraybackslash}m{0.16\textwidth}|>{\centering\arraybackslash}m{0.16\textwidth}} \toprule 
\multirow{2}{*}{\textbf{Beam Characteristics}} &
\multirow{2}{*}{\textbf{Description}} &
\textbf{Construction for $\PsiB_{\SF}$} &
\textbf{Sensing for $\SF$} &
\textbf{Sensing for $\SF$} \\ 
&&
\textbf{(MSE)}& 
\textbf{(MSE)} &
\textbf{(AP)} \\ \midrule
Proposed & $\theta_{\bR}=20^{\circ} $ and $N_{\dR}=18$ (default)   &  $\mathbf{0.0179\pm0.0008}$
 &  $\mathbf{0.0187\pm0.0025}$ &  $\mathbf{0.7378\pm0.0330}$\\ \midrule
\multirow{2}{*}{Beamwidth $\theta_{\bR}$} 
& $\theta_{\bR}=15^{\circ} $ & $0.0183\pm0.0007$ & $0.0208\pm0.0037$&  $0.7477\pm0.0394$\\ 
& $\theta_{\bR}=10^{\circ} $ & $0.0198\pm0.0023$ & $0.0251\pm0.0019$&  $0.6903\pm0.0199$\\ \midrule
\multirow{8}{*}{Number $N_{\dR}$ of beams}  
& $N_{\dR}=9$ (random)  & $0.0209\pm0.0020$ & $0.0357\pm0.0083$&  $0.5863\pm0.0553$\\
& $N_{\dR}=6$ (random) & $0.0198\pm0.0025$ & $0.0515\pm0.0195$&  $0.5215\pm0.0671$\\
& $N_{\dR}=3$ (random) & $0.0196\pm0.0011$ & $0.0883\pm0.0055$&  $0.3561\pm0.0343$\\
& $N_{\dR}=2$ (random) & $0.0247\pm0.0007$ & $0.1467\pm0.0134$&  $0.1887\pm0.0318$\\
& $N_{\dR}=9$ (regular)  & $0.0199\pm0.0011$ & $0.0326\pm0.0046$&  $0.6327\pm0.0350$\\
& $N_{\dR}=6$ (regular) & $0.0213\pm0.0009$ & $0.0461\pm0.0072$&  $0.5393\pm0.0455$\\
& $N_{\dR}=3$ (regular) & $0.0239\pm0.0009$ & $0.0827\pm0.0102$&  $0.3561\pm0.0446$\\
& $N_{\dR}=2$ (regular) & $0.0227\pm0.0008$ & $0.1280\pm0.0140$& $0.2326\pm0.0374$ \\
\bottomrule
\end{supertabular}
\end{table*}
In this section, we investigate the impact of different voting strategies on both construction and sensing performance, as illustrated in Fig.~\ref{fig:voting_eva}, where the \ac{mse} and AP are plotted against training epochs and achieved based on scenarios $\SC_{2}$. The four voting strategies, labeled $1$ through $4$, correspond to the following configurations: 
\begin{itemize}
\item Strategy 1: \( \eta = 1 - \max(0, 1 - \frac{\text{epoch}}{50}) \),
\item Strategy 2: \( \eta = 1 \), which is independent of epoch,
\item Strategy 3: No voting, i.e., \( \sum_{\theta_{d} \in \DC} C(\SF, \max(0, \hat{\SF}_{\theta_{d}} - \Psi_{\smax})) \) for training,
\item Strategy 4: \( f_{\ENS}(\{\hat{\SF}_{\theta_{d}}\}_{\theta_{d} \in \DC}) = \frac{1}{N_{\dR}}\!\! \sum_{\theta_{d} \in \DC} \max(0, \hat{\SF}_{\theta_{d}}\!\!-\!\!\Psi_{\smax}) \) for training.
\end{itemize}

\begin{figure*}[t]
    \centering
    \subfloat[MSE of THz radio
map construction task]{%
        \includegraphics[width=0.33\textwidth]{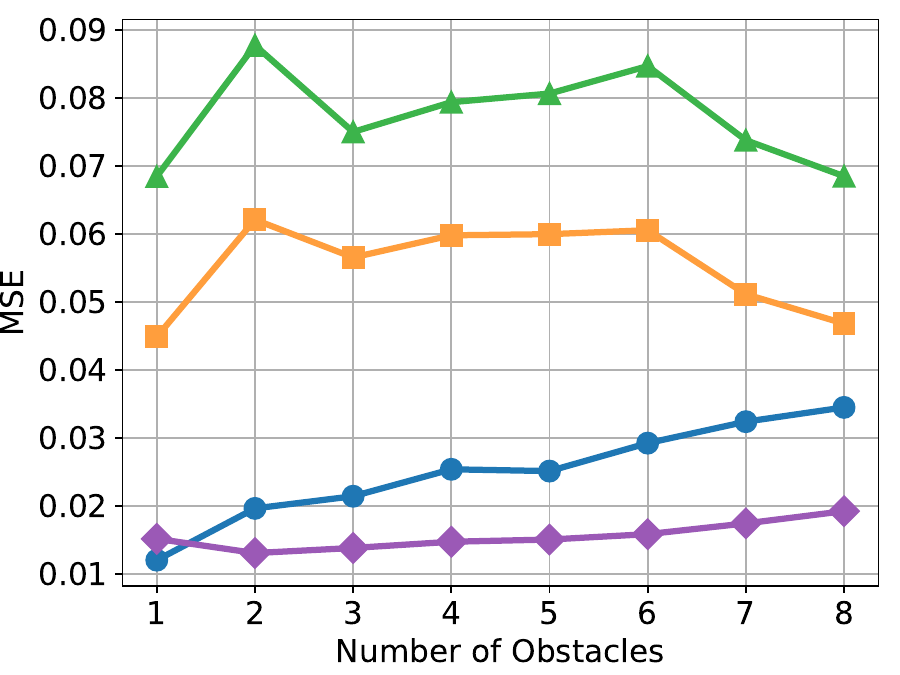}
        \label{fig:real_evaCons}
    } 
    \subfloat[MSE of obstacle sensing task]{%
\includegraphics[width=0.33\textwidth]{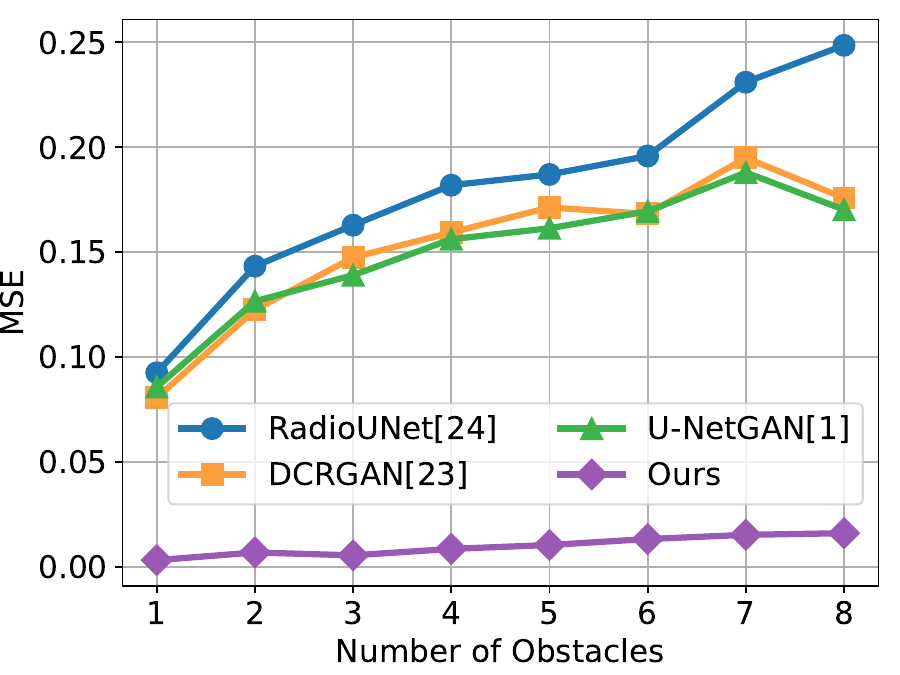}
        \label{fig:real_evaobst}
    }
        \subfloat[AP of obstacle sensing task]{%
\includegraphics[width=0.33\textwidth]{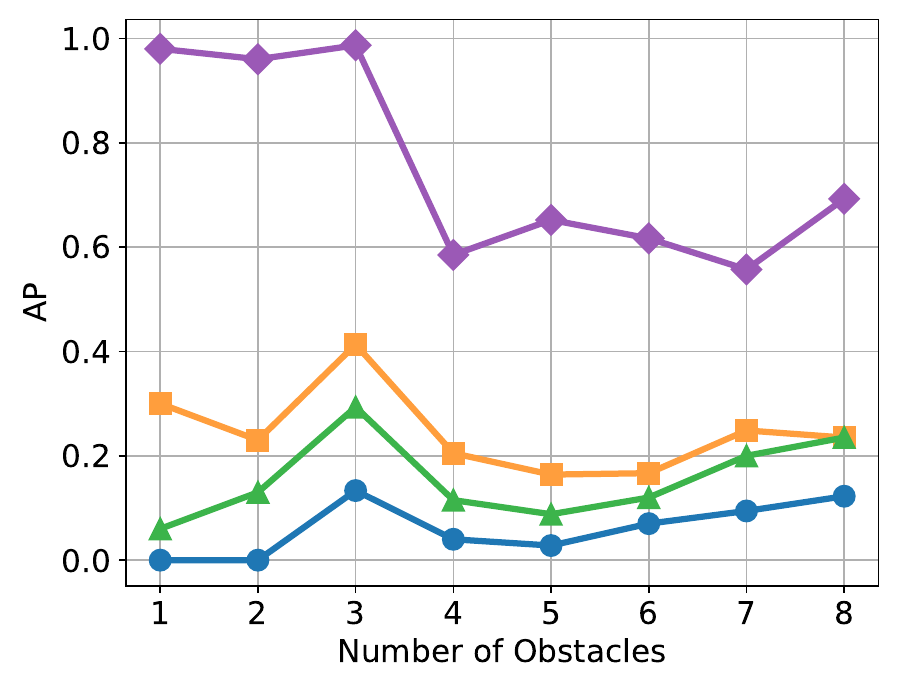}
        \label{fig:real_evaAP}
    }
    \caption{Performance of different tasks versus the number of obstacles for different approaches evaluated in real-world city scenario $\SC_3$.}
    \label{fig:real_eva}
\end{figure*}


It is evident that the proposed voting strategy achieves superior performance across both tasks, which validates the effectiveness of the design of \( \eta = \text{sign(epoch - 50)} \) for ensuring an error probability \( \epsilon < 0.5 \) and the unanimous voting enabled by equations~\eqref{hard_vote} and \eqref{soft_vote}. Moreover, the sharp changes in the \ac{mse} curves (with respect to the proposed strategy, as well as strategies 3 and 4) near the 50th epoch mark in Fig.~\ref{fig:voting_evaCons} and \ref{fig:voting_evaAP} highlight the critical role of integration in improving performance. To further understand the impact of the proposed voting-based sensing scheme with respect to \ac{thz} beam characteristics, we present the effects of beamwidth \( \theta_{\bR} \) and the number of beams \( N_{\dR} \) in Table~\ref{tab_voting}, using a fixed set \( \DC \) of 18 beam directions. It is observed that a larger beamwidth leads to better accuracy in terms of MSE, as the wider beam increases the likelihood of interactions between \ac{thz} signals and obstacles, thereby enabling the \ac{cgan} to extract more detailed obstacle information from \( \wh{\PsiB}_{\SF} \). This, in turn, improves the construction accuracy of \( \wh{\PsiB}_{\SF} \) through integration. Additionally, the results with respect to AP suggest that a moderate beamwidth may be preferable for AP-based sensing performance. 

Regarding the impact of \( N_{\dR} \), we consider both random and regular beam arrangements from \( \DC \), where the latter involves selecting \( N_{\dR} \) evenly spaced beams. Generally, a higher \( N_{\dR} \) results in more votes from \( \wh{\PsiB}_{\SF} \) and better MSE and/or AP performance in both tasks, which aligns with the observations in Fig.~\ref{GT_beam} and equation \eqref{ensemble_inequality}. Compared to using \(N_{\dR}=2\) directional THz radio maps for voting, increasing the number \( N_{\dR}\) of votes to $18$ leads to mean MSE reductions of up to 27.5\% and 87.3\% in the reconstruction and sensing tasks, respectively, while achieving a $3.91$-fold improvement in AP for sensing capability. This effect is particularly pronounced in obstacle sensing, as voting is specifically designed to improve this task, thereby aiding the proposed CGAN model in achieving a more comprehensive understanding of environmental semantics. Moreover, the regular arrangement performs better because evenly spaced beams provide more comprehensive coverage in terms of both communication and sensing.

\subsection{Performance in Real-World City Scenario}
Under the validation scenario \( \SC_{3} \) with a real-world city, we investigate the impact of the number \( N_{\IR} \) of obstacles on the proposed framework during real-world inference. Specifically, we sequentially remove obstacles from Fig.~\ref{evaluation_scenario}(b) and analyze the performance by plotting the  \ac{mse} for both construction and sensing tasks in Fig.~\ref{fig:real_evaCons} and \ref{fig:real_evaobst}, respectively, as well as the AP for sensing in Fig.~\ref{fig:real_evaAP}, all as functions of \( N_{\IR} \). Both tasks exhibit performance degradation, as reflected by an increase in MSE and/or a decrease in AP, with the rising complexity (i.e., \( N_{\IR} \)) of the scenario. Notably, the proposed framework maintains superior performance over the baselines across different  \( N_{\IR} \).

Fig.~\ref{fig:real_evaCons} and \ref{fig:real_evaobst} illustrate that the advantages of integration become more pronounced in complex scenarios, as evidenced by the widening performance gap between the proposed framework and the baseline models, further highlighting its strong generalization capability. This suggests that the proposed integration effectively mitigates the adverse effects of increased complexity, particularly in obstacle sensing with accuracy consideration. However, when considering AP, which reflects both the precision and recall of obstacle sensing, the above observations do not fully hold. As shown in Fig.~\ref{fig:real_evaAP}, when the scenario complexity exceeds a certain threshold (e.g., \( N_{\IR}>3 \)), the AP values of both the proposed framework and the baselines fluctuate within a limited range. That is, the sensing performance on AP may approach a level close to its lower limit when \( N_{\IR}>3 \). This can be attributed to AP's greater sensitivity than MSE, as AP inherently reflects the balance between precision and recall, which is particularly challenging to achieve in highly complex scenarios.

In terms of MSE for sensing, the proposed framework achieves up to a 90.6\% accuracy improvement compared to U-netGAN \cite{hu2024towards}, the best-performing method among the baselines, when the scenario $\SC_{3}$ is most complex, i.e. $N_{\IR}=8$. Meanwhile, compared to RadioUNet \cite{levie2021radiounet} for construction, our integration yields a 44.3\% accuracy improvement, which is also a noteworthy result. In contrast, when the complexity is extremely low, the advantage diminishes, and the proposed framework may even underperform compared to construction-focused baselines in \ac{thz} radio map construction.


\section{Conclusions}
\label{sec:conc} 
In this paper, we have explored communication-assisted sensing in \ac{thz} \ac{isac}, where limited \ac{rss} measurements are utilized for obstacle sensing and \ac{thz} radio map construction. To address the challenges of enhancing generalization ability and balancing accurate sensing with construction, we have formulated an integrated construction and sensing problem. To solve this problem with improved generalization and mutual benefit, we propose a \ac{cgan}-based framework coupled with a voting-based sensing scheme. Simulation results demonstrate that the proposed framework offers significant integrated benefits across various generalization scenarios, with MSE reductions of up to 44.3\% and 90.6\% for the construction and sensing tasks, respectively. The framework also achieves notable improvements in AP, indicating better precision-recall balance in sensing. As more THz radio maps from different beam directions are exploited, the environmental semantics extracted by the proposed voting scheme become increasingly accurate, achieving up to a 87.3\% accuracy improvement in MSE for sensing, which in turn leads to a 27.5\% improvement for reconstruction due to integration.


\appendices

\section{Proof of Proposition  1}
\label{Annex:a}
Based on the majority voting strategy, one can obtain 
\begin{equation}
\begin{aligned}
\relax[\hat{\SF}]_{i,j}=\sign\left(\sum_{d=1}^{N_{\dR}}[\hat{\SF}_{\theta_{d}}]_{i,j}-\frac{N_{\dR}}{2}\right),
 \label{voting}
\end{aligned}
\end{equation}  where $N_{\dR}$ is an odd number for convenience. The error probability of the final result $[\hat{\SF}]_{i,j}$ of the function $f_{\OCC}$ is 
\begin{equation}
\begin{aligned}
P\left([\hat{\SF}]_{i,j}\neq [\SF]_{i,j} \right)=&P\Big(N_{\corr} \leq \frac{N_{\dR}}{2}\Big)=P\Big(N_{\corr} \leq \lfloor\frac{N_{\dR}}{2}\rfloor\Big)\\
=&\sum_{c=0}^{\lfloor\frac{N_{\dR}}{2}\rfloor}\dbinom{N_{\dR}}{c}\left(1-\epsilon\right)^{c}\epsilon^{N_{\dR}-c},
 \label{ensemble_error}
\end{aligned}
\end{equation}where $N_{\corr}=\sum_{d=1}^{N_{\dR}}\sign\left(1-\left|[\SF]_{i,j}-[\hat{\SF}_{\theta_{d}}]_{i,j}\right|\right)$ is the number of correctly estimated votes. From the concentration form of Hoeffding's inequality \cite{hoeffding1994probability}, one can obtain 
\begin{equation}
\begin{aligned}
P\left(\frac{N_{\corr}}{N_{\dR}}-\left(1-\epsilon\right)\leq-\delta\right)\leq e^{-2 N_{\dR} \delta^{2}},
 \label{Hoeffding}
\end{aligned}
\end{equation} where $\delta>0$. Substituting $\delta=\frac{1}{2}-\epsilon$ into (\ref{Hoeffding}), then (\ref{ensemble_error}) can be written as
\begin{equation}
\begin{aligned}
P\left([\hat{\SF}]_{i,j}\neq [\SF]_{i,j} \right)\leq e^{-2 N_{\dR} \left(\frac{1}{2}-\epsilon\right)^{2}}.
 \label{ensemble_inequality_proof}
\end{aligned}
\end{equation}

\ifCLASSOPTIONcaptionsoff
  \newpage
\fi

\bibliographystyle{IEEEtran}
\bibliography{IEEEabrv,ref}

\end{document}